%% file: main-v2.tex
\documentclass[runningheads]{llncs-v2}

\usepackage{amssymb}
\usepackage{xspace}
\usepackage{amsmath}
\usepackage[utf8]{inputenc}
\usepackage[]{hyperref}
\usepackage[ruled,vlined,linesnumbered]{algorithm2e}
\usepackage{xfrac}
\usepackage{xy}
\usepackage{tikz}
    \usetikzlibrary{arrows}
    \usetikzlibrary{calc}
    \usetikzlibrary{shapes,positioning}
\usepackage{float}
\usepackage{subcaption}
\usepackage{soul}
\usepackage{xcolor}
\usepackage{wrapfig}
\usepackage{apxproof}

\newcommand\ie{\textit{i.e.}\xspace}
\newcommand\eg{\textit{e.g.}\xspace}
\newcommand\lhs{l.h.s.\xspace}
\newcommand\rhs{r.h.s.\xspace}
\newcommand\cat[1]{\ensuremath{\mathbf{#1}}}
\newcommand\Set{\ensuremath{\cat{Set}}}

\newcommand\y{\ensuremath{\mathbf{y}}}
\newcommand\GM{\ensuremath{\hat{\cat{C}}}}
\newcommand\GI{\ensuremath{\GM_\mathcal{M}}}
\newcommand\C{\cat{C}}
\newcommand\I{\cat{I}}
\newcommand\Gam{\mathbf{\Gamma}}

\newcommand\Hom[3]{\mathrm{Hom}_{#1}(#2, #3)}

\newcommand\D{\mathrm{D}}
\renewcommand\L{\mathrm{L}}
\newcommand\R{\mathrm{R}}

\newcommand\U{\mathrm{U}}

\newcommand\Colim{{\rm Colim}}
\newcommand\Proj{{\rm Proj}}
\newcommand\mto{\hookrightarrow}
\newcommand\mfrom{\hookleftarrow}
\renewcommand\to{\rightarrow}
\newcommand\from{\leftarrow}

\newcommand{\nodegraph}[2]{
    \draw[color=gray, rounded corners=3pt] (#1+0.25,#2+0.25) rectangle ++(1-0.5,1-0.5);
    \node[fill,circle,scale=0.3] (d0) at (#1 + 0.5, #2 + 0.5) {};
}

\newcommand{\nodegraphpp}[2]{
    \draw[color=gray, rounded corners=3pt] (#1+0.25,#2+0.25) rectangle ++(1-0.5,1-0.5); \node[fill,circle,scale=0.3] (dpp0) at (#1 + 0.5, #2 + 0.5) {};
}

\newcommand{\edgegraph}[2]{
    \draw[color=gray, rounded corners=5pt] (#1+0.25,#2+0.25) rectangle ++(2-0.5,1-0.5);
    \node[fill,circle,scale=0.3] (e0) at (#1+0.5,#2+0.5) {};
    \node (me0) at (#1+1,#2+0.5) {};
    \node[fill,circle,scale=0.3] (e1) at (#1+1.5,#2+0.5) {};
    \draw [shorten >=3pt,shorten <=3pt, ->] (e0) edge (e1);
}

\newcommand{\eedgegraph}[2]{
    \draw[color=gray, rounded corners=5pt] (#1+0.25,#2+0.25) rectangle ++(2-0.5,1-0.5);
    \node[fill,circle,scale=0.3] (ee0) at (#1+0.5,#2+0.5) {};
    \node[fill,circle,scale=0.3] (ee01) at (#1+1,#2+0.5) {};
    \node[fill,circle,scale=0.3] (ee1) at (#1+1.5,#2+0.5) {};
    \draw [shorten >=1pt,shorten <=1pt, ->] (ee0) edge (ee01);
    \draw [shorten >=1pt,shorten <=1pt, ->] (ee01) edge (ee1);
}

\newcommand{\eedgegraphp}[2]{
    \draw[color=gray, rounded corners=5pt] (#1+0.25,#2+0.25) rectangle ++(2-0.5,1-0.5);
    \node[fill,circle,scale=0.3] (eep0) at (#1+0.5,#2+0.5) {};
    \node[fill,circle,scale=0.3] (eep01) at (#1+1,#2+0.5) {};
    \node[fill,circle,scale=0.3] (eep1) at (#1+1.5,#2+0.5) {};
    \draw [shorten >=1pt,shorten <=1pt, ->] (eep0) edge (eep01);
    \draw [shorten >=1pt,shorten <=1pt, ->] (eep01) edge (eep1);
}

\newcommand{\emedgegraph}[2]{
    \draw[color=gray, rounded corners=5pt] (#1+0.25,#2+0.25) rectangle ++(2-0.5,1-0.5);
    \node[fill,circle,scale=0.3] (eme0) at (#1+0.5,#2+0.5) {};
    \node[fill,circle,scale=0.3] (eme01) at (#1+1,#2+0.5) {};
    \node[fill,circle,scale=0.3] (eme1) at (#1+1.5,#2+0.5) {};
    \draw [shorten >=1pt,shorten <=1pt, ->] (eme0) edge (eme01);
    \draw [shorten >=1pt,shorten <=1pt, ->] (eme1) edge (eme01);
}

\newcommand{\emedgegraphp}[2]{
    \draw[color=gray, rounded corners=5pt] (#1+0.25,#2+0.25) rectangle ++(2-0.5,1-0.5);
    \node[fill,circle,scale=0.3] (emep0) at (#1+0.5,#2+0.5) {};
    \node[fill,circle,scale=0.3] (emep01) at (#1+1,#2+0.5) {};
    \node[fill,circle,scale=0.3] (emep1) at (#1+1.5,#2+0.5) {};
    \draw [shorten >=1pt,shorten <=1pt, ->] (emep0) edge (emep01);
    \draw [shorten >=1pt,shorten <=1pt, ->] (emep1) edge (emep01);
}

\newcommand{\ebedgegraph}[2]{
    \draw[color=gray, rounded corners=5pt] (#1+0.25,#2+0.25) rectangle ++(2-0.5,1-0.5);
    \node[fill,circle,scale=0.3] (ebe0) at (#1+0.5,#2+0.5) {};
    \node[fill,circle,scale=0.3] (ebe01) at (#1+1,#2+0.5) {};
    \node[fill,circle,scale=0.3] (ebe1) at (#1+1.5,#2+0.5) {};
    \draw [shorten >=1pt,shorten <=1pt, ->] (ebe01) edge (ebe0);
    \draw [shorten >=1pt,shorten <=1pt, ->] (ebe01) edge (ebe1);
}

\newcommand{\ebedgegraphp}[2]{
    \draw[color=gray, rounded corners=5pt] (#1+0.25,#2+0.25) rectangle ++(2-0.5,1-0.5);
    \node[fill,circle,scale=0.3] (ebep0) at (#1+0.5,#2+0.5) {};
    \node[fill,circle,scale=0.3] (ebep01) at (#1+1,#2+0.5) {};
    \node[fill,circle,scale=0.3] (ebep1) at (#1+1.5,#2+0.5) {};
    \draw [shorten >=1pt,shorten <=1pt, ->] (ebep01) edge (ebep0);
    \draw [shorten >=1pt,shorten <=1pt, ->] (ebep01) edge (ebep1);
}

\newcommand{\nodegraphp}[2]{
    \draw[color=gray, rounded corners=3pt] (#1+0.25,#2+0.25) rectangle ++(1-0.5,1-0.5);    \node[fill,circle,scale=0.3] (dp0) at (#1 + 0.5, #2 + 0.5) {};
}
\newcommand{\edgegraphp}[2]{
    \draw[color=gray, rounded corners=5pt] (#1+0.25,#2+0.25) rectangle ++(2-0.5,1-0.5);
    \node[fill,circle,scale=0.3] (ep0) at (#1+0.5,#2+0.5) {};
    \node (mep0) at (#1+1,#2+0.5) {};
    \node[fill,circle,scale=0.3] (ep1) at (#1+1.5,#2+0.5) {};
    \draw [shorten >=3pt,shorten <=3pt, ->] (ep0) edge (ep1);
}
\newcommand{\edgegraphpp}[2]{
    \draw[color=gray, rounded corners=5pt] (#1+0.25,#2+0.25) rectangle ++(2-0.5,1-0.5);
    \node[fill,circle,scale=0.3] (epp0) at (#1+0.5,#2+0.5) {};
    \node (mepp0) at (#1+1,#2+0.5) {};
    \node[fill,circle,scale=0.3] (epp1) at (#1+1.5,#2+0.5) {};
    \draw [shorten >=3pt,shorten <=3pt, ->] (epp0) edge (epp1);
}
\newcommand{\paredgegraph}[2]{
    \draw[color=gray, rounded corners=5pt] (#1+0.25,#2+0.25) rectangle ++(2-0.5,1-0.5);
    \node[fill,circle,scale=0.3] (ee0) at (#1+0.5,#2+0.5) {};
    \node (mee0) at (#1+1,#2+0.5) {};
    \node[fill,circle,scale=0.3] (ee1) at (#1+1.5,#2+0.5) {};
    \draw [shorten >=3pt,bend left=20,shorten <=3pt, ->] (ee0) edge (ee1);
    \draw [shorten >=3pt,bend right=20,shorten <=3pt, ->] (ee0) edge (ee1);
}
\newcommand{\paredgegraphp}[2]{
    \draw[color=gray, rounded corners=5pt] (#1+0.25,#2+0.25) rectangle ++(2-0.5,1-0.5);
    \node[fill,circle,scale=0.3] (eep0) at (#1+0.5,#2+0.5) {};
    \node (meep0) at (#1+1,#2+0.5) {};
    \node[fill,circle,scale=0.3] (eep1) at (#1+1.5,#2+0.5) {};
    \draw [shorten >=3pt,bend left=20,shorten <=3pt, ->] (eep0) edge (eep1);
    \draw [shorten >=3pt,bend right=20,shorten <=3pt, ->] (eep0) edge (eep1);
}
\newcommand{\loopedgegraph}[2]{
    \draw[color=gray, rounded corners=5pt] (#1+0.25,#2+0.25) rectangle ++(2-0.5,1-0.5);
    \node[fill,circle,scale=0.3] (le0) at (#1+0.5,#2+0.5) {};
    \node (mle0) at (#1+1,#2+0.5) {};
    \node[fill,circle,scale=0.3] (le1) at (#1+1.5,#2+0.5) {};
    \draw [shorten >=3pt,bend left=20,shorten <=3pt, ->] (le0) edge (le1);
    \draw [shorten >=3pt,bend left=20,shorten <=3pt, ->] (le1) edge (le0);
}
\newcommand{\loopedgegraphp}[2]{
    \draw[color=gray, rounded corners=5pt] (#1+0.25,#2+0.25) rectangle ++(2-0.5,1-0.5);
    \node[fill,circle,scale=0.3] (lep0) at (#1+0.5,#2+0.5) {};
    \node (mlep0) at (#1+1,#2+0.5) {};
    \node[fill,circle,scale=0.3] (lep1) at (#1+1.5,#2+0.5) {};
    \draw [shorten >=3pt,bend left=20,shorten <=3pt, ->] (lep0) edge (lep1);
    \draw [shorten >=3pt,bend left=20,shorten <=3pt, ->] (lep1) edge (lep0);
}

    \definecolor{junglegreen}{rgb}{0.16, 0.67, 0.53}

\begin{document}

\title{Accretive Computation of Global Transformations\\
{\small Extended Version}}

\titlerunning{Accretive Computation of Global Transformations - Extended Version}

\author{Alexandre Fernandez \and Luidnel Maignan \and Antoine Spicher}

\authorrunning{A. Fernandez et al.}

\institute{Univ Paris Est Creteil, LACL, 94000, Creteil, France\\\email{firstname.lastname@u-pec.fr}}

\maketitle

\begin{abstract}
Global transformations form a categorical framework adapting graph transformations to describe fully synchronous rule systems on a given data structure.
In this work we focus on data structures that can be captured as presheaves and study the computational aspects of such synchronous rule systems.
To obtain an online algorithm, a complete study of the sub-steps within each synchronous step is done at the semantic level.
This leads to the definition of accretive rule systems and a local criterion to characterize these systems.
Finally an online computation algorithm for theses systems is given.
\end{abstract}

\newtheoremrep{proposition}[theorem]{Proposition}
\newtheoremrep{lemma}[theorem]{Lemma}
\newtheoremrep{theorm}[theorem]{Theorem}

\section{Introduction}

Classically, a graph rewriting system consists of a set of rewriting rules $l \Rightarrow r$ expressing that $l$ should be replaced by $r$ somewhere in an input graph.
Usually rules are applied one after the other in a non-deterministic way~\cite{corradini2006sesqui,ehrig2006graph,ehrig1973graph}.
Allowing multiple rules to be applied simultaneously has been the subject of multiple studies, leading to the concepts of parallel rule applications, concurrent rule applications~\cite{ehrig1980parallelism}, and amalgamation of rules~\cite{boehm1987amalgamation}.
For instance, amalgamation of rules is considered when two rules $l \Rightarrow r$ and $l' \Rightarrow r'$ are applicable but $l$ and $l'$ overlap.
Basically, the behavior on the overlap is given by a third rule specifying how $r$ and $r'$ should consequently overlap.
But some systems do not only require the amalgamation of a few, finite, number of rule applications, but the amalgamation of an unbounded number, the whole input being transformed.
A simple example is triangular mesh refinement where the triangles of a mesh are all subdivided into many smaller triangles simultaneously, with a coherent behavior on the overlap between triangles~\cite{maignan2015global}.
In this extreme case, the notion of replacement is not appropriate, no part of the initial mesh is really kept identical.

Rethinking rewriting for those particular systems where the transformation is globally coherent leads to a generic and economical mathematical structure captured easily with categorical concepts, the so-called \emph{global transformations}~\cite{maignan2015global}.
This point of view has been applied mathematically to examples like mesh refinements on abstract cell complexes~\cite{maignan2015global}, but also deterministic Lindenmayer systems acting on formal words~\cite{fernandez2019lindenmayer}, and cellular automata acting on labeled Caley graphs~\cite{DBLP:conf/automata/FernandezMS21}.
In the present work, we tackle global transformations in an algorithmic perspective and show how they can be computed in an online fashion when transforming graphs, but also any generalization of graphs suitably captured by categories of presheaves (labeled graphs, higher-dimensional graphs, etc.).
This online strategy saves memory during the computation, more memory being also saved through a condition allowing the modifications to happen \emph{in place}: accretiveness.

The article is organized as follows.
After adapting in Section~\ref{sec:backgrounds} the definition of global transformations to presheaves, 
Section~\ref{sec:online} unfolds all implications of the online and accretive perspective at the semantic level, and gathers all necessary formal results.
This leads to the presentation of the algorithm in Section~\ref{sec:algo}, followed by a discussion in Section~\ref{sec:conclusion}.
In the present version, facts are only stated.

\section{Background on Global Transformations}
\label{sec:backgrounds}

In the section, we adapt the definitions of global transformations given in~\cite{fernandez2019lindenmayer,maignan2015global} to fit with the context of presheaves and monomorphisms between them.
The reader is assumed to be familiar with the definitions of categories, functors, monomorphisms, comma categories, diagrams, cocones, colimits and categories of presheaves.
Refer to~\cite{mac2013categories} for details.
These constructions are also pedagogically introduced in the context of global transformation in~\cite{fernandez2019lindenmayer}.

In the following, we consider an arbitrary category $\C$ and write $\GM$ for the category $\Set^{\C^\mathbf{op}}$ of all presheaves on $\C$, $\GI$ for the subcategory restricting morphisms to monomorphisms, and $\U : \GI \to \GM$ for the obvious forgetful functor.
Morphisms of $\GI$ and monomorphisms of $\GM$ are written $p \mto p'$.
We write $\y: \C \to \GM$ for the Yoneda embedding, and call \emph{representable presheaves} the image $\y{c}$ of any $c \in \C$.

The examples are spelled out with $\C$ set to the category with two objects $\mathtt{v}$ and $\mathtt{e}$, and two morphisms $\mathtt{s}, \mathtt{t} : \mathtt{v} \to \mathtt{e}$.
A presheaf $p \in \GM$ is then a directed multigraph with self-loops: $p(\mathtt{v})$ and $p(\mathtt{e})$ are respectively the sets of vertices and edges composing the graph, and $p(\mathtt{s})$ (resp. $p(\mathtt{t})$) is a function mapping each edge to its source (resp. target).
The representable presheaves are the graph $\y{\mathtt{v}}$ with a single vertex and the graph $\y{\mathtt{e}}$ with two vertices and a single edge.
We will make use of the following particular graphs: $d_k$ the discrete graph with $k$ vertices and no edge, $p_k$ the path of length $k$, and $c_k$ the cycle of length $k$, $k > 0$.

The category $\GM$ is cocomplete and for any diagram $D : \I \to \GM$, the colimit $C$ of $D$ is directly written $\Colim(D)$; $C$ also abusively designates the apex and $C_i: D(i) \to C$ the cocone components for any $i \in \I$.
The usual description of colimits in $\GM$ based on equivalence classes of vertices and edges can be rephrased in terms of \emph{zig-zag}.
A \emph{zig-zag} $z$ in a category $\cat{X}$ is given by some natural number $|z|$, a sequence $\langle z_i \rangle_{0 \leq i \leq |z|}$ of $|z|+1$ objects of $\cat{X}$ and a sequence $\langle \overline{z}_i \rangle_{0 \leq i \leq |z|-1}$ of morphisms in $\cat{X}$ of the form $z_0 \to z_1 \from z_2 \to z_3 \cdots z_{|z|}$ or $z_0 \from z_1 \to z_2 \from z_3 \cdots z_{|z|}$.
Given a functor $F: \cat{X} \to \cat{Y}$ and two morphisms $g : c \to F(z_0)$ and $g' : c \to F(z_{|z|})$ in $\cat{Y}$, we write $F(z)$ for the zig-zag defined with $|F(z)| = |z|$, $F(z)_i = F(z_i)$, and $\overline{F(z)}_i = F(\overline{z}_i)$.
Given two morphisms $f_0 : c \to z_0$ and $f_{|z|} : c \to z_{|z|}$ in $\cat{Y}$, $z$ is said to \emph{link $f_0$ and $f_{|z|}$} if there is a sequence $\langle f_i : c \to z_i \rangle_{1\le i\le|z|-1}$ of morphisms such that, for any $i \in \{0,\ldots,|z|-1\}$, $f_i = \overline{z}_i \circ f_{i+1}$ or $\overline{z}_i \circ f_i = f_{i+1}$ depending on the direction of $\overline{z}_i$.
We say that \emph{$z$ links $g$ and $g'$ through F} if the zig-zag $F(z)$ links $g$ and $g'$.

\begin{proposition}\label{prop:colim-graphi}
For any diagram $D : \I \to \GM$ of small domain $\I$ with $C = \Colim(D)$, any representable presheaf $\y{c}$ and any $x : \y{c} \to C$, there is at least one pair of $\langle i \in \I, y : \y{c} \to D(i)\rangle$ such that $x = C_i \circ y$, and any two such pairs $\langle i, y \rangle$ and $\langle i', y' \rangle$ have a zig-zag $z$ in $\I$ that links $y$ and $y'$ through $D$.

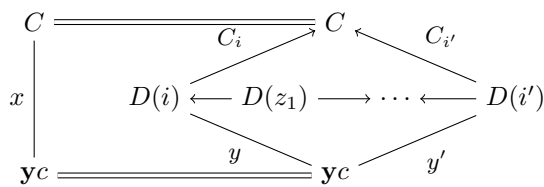
\begin{center}
\begin{tikzpicture}[yscale=0.6,xscale=0.8]
    \node (c) at (3,1.7) {$C$};
    \node (z0) at (0,0) {$D(i)$};
    \node (z1) at (2,-0) {$D(z_1)$};
    \node (z2) at (4,0) {$\ldots$};
    \node (zz) at (6,-0) {$D(i')$};
    \node (a) at (3,-1.7) {$\y{c}$};

    \node (c2) at (-2,1.7) {$C$};
    \node (a2) at (-2,-1.7) {$\y{c}$};

    \draw[->] (z1) to  (z0);
    \draw[->] (z1) to  (z2);
    \draw[->] (zz) to  (z2);
    \draw[left] (a) to node[below left,font=\small]{$y$} (z0);
    \draw[right] (a) to node[below right,font=\small]{$y'$} (zz);
    \draw[->] (z0) to  node[above left,font=\small]{$C_i$} (c);
    \draw[->] (zz) to node[above right,font=\small]{$C_{i'}$} (c);

    \draw[right] (a2) to node[left]{$x$} (c2);
    \draw[double, double distance=.5mm] (a) to (a2);
    \draw[double, double distance=.5mm] (c) to (c2);

\end{tikzpicture}
\end{center}

\end{proposition}
\begin{proof}
    Colimits in categories of presheaves are computed pointwise, and colimits in $\cat{Set}$ can be described as the connected components in the category of elements of the diagram, such connected components being such sets of pairs $\langle i, x \in F(i) \rangle$ connected by zig-zags.
    Moreover, the existence of such a pair is ensured since the colimit cocone is jointly epimorphic.
    Combined with the fact that, for any presheaf $p$, elements in the set $p(c)$ are in bijection with morphisms $\y{c} \to p$ by the Yoneda lemma, we obtain this formulation.
\qed
\end{proof}

Given two categories $A$ and $B$, a functor $F : A \to B$, and an object $b$ in $B$, the comma category $F/b$ sees its objects described as pairs $\langle a \in A, f : F(a) \to b \rangle$ and its morphisms from $\langle a, f \rangle$ to $\langle a', f' \rangle$ as pairs $\langle e : a \to a', f' \rangle$ such that $f = f' \circ F(e)$.
The composition of $\langle e', f'' \rangle \circ \langle e, f' \rangle$ is therefore $\langle e' \circ e, f'' \rangle$.


\paragraph*{Specification of Global Transformations.}
\begin{figure}[t]
\begin{center}
\input{sierpinsky-1-v2.tex}
\caption{Sierpinski rule system: one rule to divide the relevant triangles, the two others and their inclusions to specify the connections in the output based on the connections in the input.}
\label{fig:sierpinskyRules}
\end{center}
\end{figure}
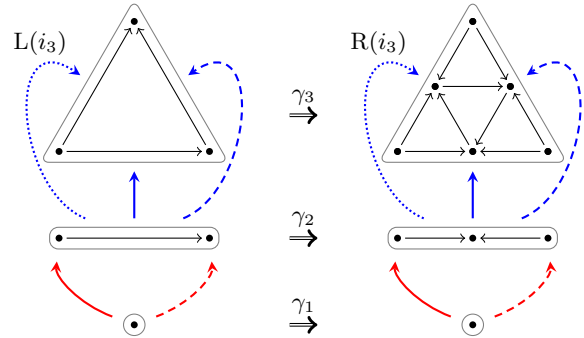
In this paper, we restrict ourselves to global transformations acting on $\GI$.
At a basic level, they are rewriting systems transforming presheaves into presheaves.
As such, their specification is based on a set of rules.
Each rule $\gamma$ is a pair written $l \Rightarrow r$ with $l, r \in \GI$.
Given an input presheaf $p$, it expresses that any occurrence of the left hand side (\lhs) $l$ in $p$ produces the corresponding right hand side (\rhs) $r$ in the associated output.
The main feature of global transformations is to endow this set of rules with a structure of \emph{category} where morphisms describe \emph{inclusions of rules}.
A rule inclusion $i: \gamma_1 \to \gamma_2$ from a \emph{sub-rule} $\gamma_1 = l_1 \Rightarrow r_1$, to a \emph{super-rule} $\gamma_2 = l_2 \Rightarrow r_2$ expresses how an occurrence of $l_1$ in $l_2$ is locally transformed into an occurrence of $r_1$ in $r_2$.
So a rule inclusion $i$ is a pair $\langle i_l:l_1 \to l_2, i_r: r_1 \to r_2 \rangle$.
Formally, such a presentation is captured by a category and two functors.

\begin{definition}\label{def:rule-system}
A \emph{rule system $T$ on $\GI$} is a tuple $\langle \Gam_T, \L_T, \R_T \rangle$ where $\Gam_T$ is a category whose objects are called \emph{rules} and morphisms are called \emph{rule inclusions}, $\L_T: \Gam_T \to \GI$ is a full embedding functor called the \emph{\lhs functor}, and $\R_T: \Gam_T \to \GI$ is a functor called the \emph{\rhs functor}.
The subscript $T$ is omitted when this does not lead to any confusion.
\end{definition}

Figure~\ref{fig:sierpinskyRules} illustrates a global transformation specification for generating a Sierpinski gasket.
The rule system is composed of 3 rules transforming locally vertices ($\gamma_1$), edges ($\gamma_2$) and acyclic triangles ($\gamma_3$).
These rules are related by 5 main rule inclusions: $i_1: \gamma_1 \to \gamma_2$ (plain red), $i_2: \gamma_1 \to \gamma_2$ (dashed red), $i_3: \gamma_2 \to \gamma_3$ (dotted blue), $i_4: \gamma_2 \to \gamma_3$ (plain blue), $i_5: \gamma_2 \to \gamma_3$ (dashed blue).
For instance, consider the inclusion $i_3$ which expresses that the left edge of triangle $\L(\gamma_3)$ is transformed into the left double-edge of $\R(\gamma_3)$.
Formally, this is specified via the inclusion $i_3$ whose both components $\L(i_3)$ and $\R(i_3)$ are depicted in dotted blue arrows. 
The reader is invited to pay attention that even if Fig.~\ref{fig:sierpinskyRules} does not show them, the category $\Gam$ also contains compositions of the 5 main rule inclusions (\eg $i_3 \circ i_1$), identities and symmetries, that, as functors, $\L$ and $\R$ do respect.


\paragraph*{Computing with Global Transformations.}

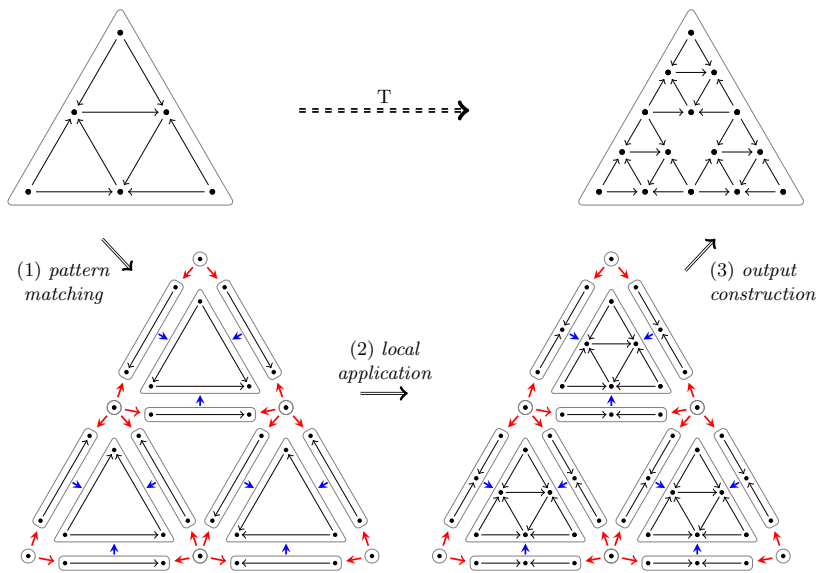
\begin{figure}[t]
\begin{center}
    \resizebox{\linewidth}{!}{\input{sierpinsky-2-v2.tex}}
    \caption{Step of computation of the Sierpinski gasket using the rules of Fig.~\ref{fig:sierpinskyRules}.}
    \label{fig:sierpinskyStep}
\end{center}
\end{figure}

Given a rule system $T$, its application on an arbitrary presheaf $p$ is a three-step process.
An illustration is given Fig.~\ref{fig:sierpinskyStep} based on the rule system of Fig.~\ref{fig:sierpinskyRules}.
The input is depicted at the top left and the output at the top right.
\begin{enumerate}

    \item \emph{Pattern matching} which consists in decomposing the input presheaf by mean of the rule \lhs
    It results a collection of \lhs instances, also called matches, structured by rule inclusions.
    This step is achieved by considering the comma category $\L_T/p$: objects in that category are indeed all the morphisms from some \lhs to $p$; morphisms are the instantiations of the rule inclusions between those matches.
    See arrow (1) in Fig.~\ref{fig:sierpinskyStep} for an illustration.
    Formally, the figure at bottom left is a representation of $\L_T \circ \Proj[\L_T/p]$ where $\Proj$ designates the first projection of the comma category mapping each instance $\langle \gamma \in \Gam_T, f: \L_T(\gamma) \mto p \rangle$ to the used rule $\gamma$.
    Notice the role of the rule inclusions (in red and blue) which are reminiscent of the input structure.
    
    \item \emph{Local application of rules} which consists in locally transforming each found \lhs into its corresponding \rhs, the structure being conserved thanks to rule inclusions.
    This step is achieved by applying the \rhs functor $\R_T$ on each rule instance: $\R_T \circ \Proj[\L_T/p]$, as illustrated in Fig.~\ref{fig:sierpinskyStep}.

    \item \emph{Output construction} which consists in assembling the output from the structured collection of \rhs
    The inclusions take here their full meaning as they are used to align the \rhs and drive the merge.
    See arrow (3) in Fig.~\ref{fig:sierpinskyStep} for an illustration.
    The resulting presheaf is formally the apex of a cocone from the diagram defined in the previous step which we used to obtain by colimit~\cite{fernandez2019lindenmayer,maignan2015global}.
    Since colimits are only guaranteed in $\GM$, we consider the following functor $\overline{T} : \GI \to \GM$:
    \begin{equation}
    \label{eq:gtstep}
    \overline{T}(\--) = \Colim(\D_T(\--))\qquad\textnormal{with}\qquad\D_T(\--) = \U \circ \R_T \circ \Proj[\L_T/\--]
    \end{equation}
    using the forgetful functor $\U$, $\overline{T}(p)$ being the result of the application.
\end{enumerate}

\begin{remark}
\label{rem:overlineT}
Notice that $\overline{T}$ is a complete functor also acting on morphisms.
Consider a monomorphism $h: p \mto p'$.
By definition of colimits, $\overline{T}(p)$ is the universal cocone with components $\overline{T}(p)_{\langle \gamma, f \rangle}: \D_T(p)(\langle \gamma, f \rangle) \to \overline{T}(p)$ for each instance $\langle \gamma, f \rangle \in \L_T/p$.
We have a similar construction for $\overline{T}(p')$ which gives rise to a cocone $C$ as the restriction of $\overline{T}(p')$ on the diagram of $\overline{T}(p)$.
Formally, $C$ is defined with apex $C = \overline{T}(p')$ and components $C_{\langle \gamma, f \rangle} = \overline{T}(p')_{\langle \gamma, h \circ f \rangle}$.
The image $\overline{T}(h)$ is the mediating morphism from colimit $\overline{T}(p)$ to $C$.
\end{remark}

We focus on those rule systems where the results stay inside $\GI$, \ie, such that all previous mediating morphisms are monomorphisms.
This leads to the following definition of global transformation for $\GI$.

\begin{definition}
A \emph{global transformation} $T$ is a rule system such that $\overline{T}$ factors through the forgetful functor $\U : \GI \to \GM$.
In this case, we denote $T: \GI \to \GI$ the functor such that
$
\U \circ T = \overline{T}.
$
\end{definition}

The Sierpenski rule system of Fig.~\ref{fig:sierpinskyRules} is a global transformation.
Its behavior is to split all edges and add some edges for acylic triangles.
Thus, adding vertices and edges to an input only adds vertices and edges to the output.
The induced functor maps monomorphisms to monomorphisms, so factors through $\U$.

Figure~\ref{fig:figex} introduces four additional examples of rule systems that illustrate the various properties that we consider exhaustively.
Let us see which of them are global transformations as a preparation for later considerations.

\newcommand{\rulesep}{\unskip\ \vrule\ }
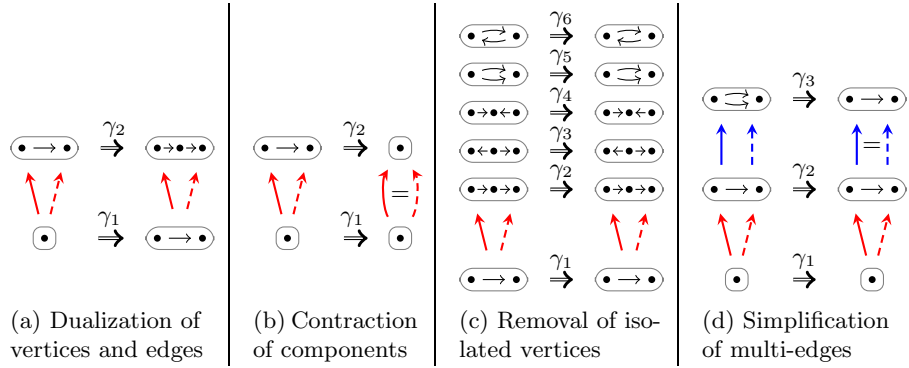
\begin{figure}[t]
\hfill
\begin{subfigure}[t]{0.22\textwidth}
\begin{tikzpicture}[scale=0.6]
    \node (u) at (1, -0.5) {};
    \begin{scope}
        \edgegraph{0}{2}
        \nodegraph{0.5}{0}
        \draw [color=red, >=stealth, thick, shorten >=8pt, shorten <=8pt, ->] (d0) edge (e0);
        \draw [color=red, densely dashed, >=stealth, thick, shorten >=8pt,shorten <=8pt, ->] (d0) edge (e1);
    \end{scope}
    \begin{scope}[shift={(3,0)}]
        \eedgegraphp{0}{2}
        \edgegraphp{0}{0}

        \draw [color=red, >=stealth, thick, shorten >=8pt,shorten <=8pt, ->] (mep0) edge (eep0);
        \draw [color=red, densely dashed, >=stealth, thick, shorten >=8pt,shorten <=8pt, ->] (mep0) edge (eep1);
    \end{scope}
    \draw (d0) edge[shorten >=18pt, shorten <=19pt, double, ->] node[above,yshift=0.1em]{$\gamma_1$} (mep0);
    \draw (me0) edge[shorten >=20pt, shorten <=18pt, double, ->] node[above,yshift=0.1em]{$\gamma_2$} (eep01);
\end{tikzpicture}
\caption{\label{fig:figexa}Dualization of\\vertices and edges}
\end{subfigure}
\hfill\rulesep\hfill
\begin{subfigure}[t]{0.18\textwidth}
\begin{tikzpicture}[scale=0.6]
    \node (u) at (1, -0.5) {};
    \begin{scope}
        \edgegraph{0}{2}
        \nodegraph{0.5}{0}
        \draw [color=red, >=stealth, thick, shorten >=8pt, shorten <=8pt, ->] (d0) edge (e0);
        \draw [color=red, densely dashed, >=stealth, thick, shorten >=8pt,shorten <=8pt, ->] (d0) edge (e1);
    \end{scope}
    \begin{scope}[shift={(2.5,0)}]
        \nodegraphpp{0.5}{2}
        \nodegraphp{0.5}{0}
        \node (=) at (1, 1.45) {$=$};

        
        \draw [color=red, >=stealth, thick, shorten >=8pt,bend left, shorten <=8pt, ->] (dp0) edge (dpp0);
        \draw [color=red, densely dashed, >=stealth, thick, shorten >=8pt,bend right, shorten <=8pt, ->] (dp0) edge (dpp0);
    \end{scope}
    \draw (d0) edge[shorten >=11pt, shorten <=20pt, double, ->] node[above,xshift=.35em, yshift=0.1em]{$\gamma_1$} (dp0);
    \draw (me0) edge[shorten >=11pt, shorten <=18pt, double, ->] node[above,xshift=.35em, yshift=0.1em]{$\gamma_2$} (dpp0);
\end{tikzpicture}
\caption{\label{fig:figexb}Contraction\\of components}
\end{subfigure}
\hfill\rulesep\hfill
\begin{subfigure}[t]{0.22\textwidth}
\begin{tikzpicture}[scale=0.6]
    \begin{scope}
        \loopedgegraph{0}{5.4}
        \paredgegraph{0}{4.55}
        \emedgegraph{0}{3.7}
        \ebedgegraph{0}{2.85}
        \eedgegraph{0}{2}
        \edgegraph{0}{0}
        \draw [color=red, >=stealth, thick, shorten >=8pt, shorten <=8pt, ->] (me0) edge (ee0);
        \draw [color=red, densely dashed, >=stealth, thick, shorten >=8pt,shorten <=8pt, ->] (me0) edge (ee1);
    \end{scope}
    \begin{scope}[shift={(3,0)}]
        \loopedgegraphp{0}{5.4}
        \paredgegraphp{0}{4.55}
        \emedgegraphp{0}{3.7}
        \ebedgegraphp{0}{2.85}
        \eedgegraphp{0}{2}
        \edgegraphp{0}{0}

        \draw [color=red, >=stealth, thick, shorten >=8pt,shorten <=8pt, ->] (mep0) edge (eep0);
        \draw [color=red, densely dashed, >=stealth, thick, shorten >=8pt,shorten <=8pt, ->] (mep0) edge (eep1);
    \end{scope}
    \draw (me0) edge[shorten >=18pt, shorten <=18pt, double, ->] node[above,yshift=0.05em]{$\gamma_1$} (mep0);
    \draw (ee01) edge[shorten >=20pt, shorten <=20pt, double, ->] node[above,yshift=0.05em]{$\gamma_2$} (eep01);
    \draw (ebe01) edge[shorten >=20pt, shorten <=20pt, double, ->] node[above,yshift=0.05em]{$\gamma_3$} (ebep01);
    \draw (eme01) edge[shorten >=20pt, shorten <=20pt, double, ->] node[above,yshift=0.05em]{$\gamma_4$} (emep01);
    \draw (mee0) edge[shorten >=18pt, shorten <=18pt, double, ->] node[above,yshift=0.05em]{$\gamma_5$} (meep0);
    \draw (mle0) edge[shorten >=18pt, shorten <=18pt, double, ->] node[above,yshift=0.05em]{$\gamma_6$} (mlep0);
\end{tikzpicture}
\caption{\label{fig:figexc}Removal of isolated vertices}
\end{subfigure}
\hfill\rulesep\hfill
\begin{subfigure}[t]{0.22\textwidth}
\begin{tikzpicture}[scale=0.6]
    \begin{scope}
        \paredgegraph{0}{4}
        \edgegraph{0}{2}
        \nodegraph{0.5}{0}
        \draw [color=blue, >=stealth, thick, shorten >=6pt,transform canvas={xshift=-1.5ex},shorten <=6pt, ->] (me0) edge (mee0);
        \draw [color=blue, densely dashed, >=stealth, thick, shorten >=6pt,transform canvas={xshift=1.5ex},shorten <=6pt, ->] (me0) edge (mee0);
        \draw [color=red, >=stealth, thick, shorten >=8pt, shorten <=8pt, ->] (d0) edge (e0);
        \draw [color=red, densely dashed, >=stealth, thick, shorten >=8pt,shorten <=8pt, ->] (d0) edge (e1);
    \end{scope}
    \begin{scope}[shift={(3,0)}]
        \edgegraphpp{0}{4}
        \edgegraphp{0}{2}
        \nodegraphp{0.5}{0}
        \node (=) at (1, 3.45) {$=$};
        \draw [color=blue, >=stealth, thick, shorten >=6,transform canvas={xshift=-1.5ex},shorten <=6pt, ->] (mep0) edge (mepp0);
        \draw [color=blue, densely dashed, >=stealth, thick, shorten >=6pt,transform canvas={xshift=1.5ex},shorten <=6pt, ->] (mep0) edge (mepp0);
        
        
        \draw [color=red, >=stealth, thick, shorten >=8pt,shorten <=8pt, ->] (dp0) edge (ep0);
        \draw [color=red, densely dashed, >=stealth, thick, shorten >=8pt,shorten <=8pt, ->] (dp0) edge (ep1);
    \end{scope}
    \draw (d0) edge[shorten >=20pt, shorten <=20pt, double, ->] node[above,yshift=0.1em]{$\gamma_1$} (dp0);
    \draw (me0) edge[shorten >=18pt, shorten <=18pt, double, ->] node[above,yshift=0.1em]{$\gamma_2$} (mep0);
    \draw (mee0) edge[shorten >=18pt, shorten <=18pt, double, ->] node[above,yshift=0.1em]{$\gamma_3$} (mepp0);
\end{tikzpicture}
\caption{\label{fig:figexd}Simplification\\of multi-edges}
\end{subfigure}
\hfill
\caption{\label{fig:figex} Some rule-systems. Note that all of them remove self-loops.}
\end{figure}

\begin{example}\label{ex:exc-gt}
    The removal of isolated vertices (Fig.~\ref{fig:figexc}) is a global transformation since removal is definitely a functorial behavior from $\GI$ to $\GI$.
\end{example}

\begin{example}\label{ex:exd-gt}
    Simplification of multi-edges (Fig.~\ref{fig:figexd}) is also a global transformation.
    For any two vertices $a$ and $b$ with at least an edge from $a$ to $b$, it merges all edges from $a$ to $b$ into a single edge.
\end{example}

\begin{example}\label{ex:exa-gt}
    On the contrary, the dualization of vertices and edges (Fig.~\ref{fig:figexa}) is not a global transformation.
    Indeed, consider a monomorphism $i: p_2 \mto c_3$ from the path of length 2 $p_2$ to the cycle of length 3 $c_3$.
    In this case $\overline{T}(p_2) = p_3$, $\overline{T}(c_3) = c_3$, and there is no monomorphism sending the four vertices of $p_3$ to the three vertices of $c_3$.
\end{example}

\begin{example}\label{ex:exb-gt}
    Contraction of components (Fig.~\ref{fig:figexb}) is not a global transformation.
    Consider the monomorphism $i: d_2 \mto p_1$ from the graph $d_2$ with only vertices to the path $p_1$ of length 1.
    In this case, $\overline{T}(d_2) = d_2$ and $\overline{T}(p_1) = d_1$ and there is no monomorphism sending the two vertices of $d_2$ to the single vertex of $d_1$.
\end{example}

\section{Accretion and Incrementality}
\label{sec:online}

We are interested in having an online algorithm computing Eq.~\eqref{eq:gtstep} where the output $\overline{T}(p)$, for any $p$, is built while the comma category $\L_T/p$ is discovered by pattern matching.
Informally, starting from a seed corresponding to the \rhs of some initial instance, $\L_T/p$ is visited from neighbor to neighbor, each instance providing a new piece of material to accumulate to the current intermediate result.
We first give the formal features to be able to speak about intermediate results, leading to the notion of \emph{accretive rule system}.
Then we give a criterion, called \emph{incrementality}, for a rule system to be an accretive global transformation.

In this algorithmic perspective, we restrict our discussion to finite presheaves and finite rule systems so that $\L_T/p$ is finite as well.
Moreover we assume that $\L_T/p$ is connected; disconnected components can be processed independently.
Finally, we fix a given finite rule system $T = \langle \Gam, \L, \R \rangle$ and a finite presheaf $p$.

\subsection{Accretive Rule Systems and Global Transformations}
\label{sec:accretive}

Informally an intermediate result consists in the application of $\overline{T}$ on an incomplete knowledge of $\L/p$, \ie, on a partial decomposition of the input.
Our study of partial decompositions starts with some remarks.
To begin, the comma category $\L/p$ is a preordered set.
\begin{proposition}
\label{prop:comma_thin}
For any category $\I$, any full embedding $F : \I \to \GI$, and any presheaf $p \in \GI$, the comma category $F/p$ is \emph{thin}, \ie, there is at most one morphism between any two objects.
\end{proposition}
\begin{proof}
For any $i_1, i_2 \in \I$, consider $\langle i_1, m_1: F(i_1) \mto p \rangle$ and $\langle i_2, m_2: F(i_2) \mto p \rangle$ of $F/p$ together with two parallel morphisms $\langle f: i_1 \to i_2, m_2 \rangle$ and $\langle g: i_1 \to i_2, m_2 \rangle$ between them.
We want to prove that $f = g$.
By definition of morphisms in $F/p$, $m_2 \circ F(f) = m_1 = m_2 \circ F(g)$.
Since $m_2$ is a monomorphism, we get $F(f) = F(g)$.
Since $F$ is in particular faithful, we obtain the desired equality $f = g$.
\end{proof}
Therefore thinking in terms of maximal, non-maximal and minimal instances, sub- and super-instances actually makes sense for any comma category $\L/p$.

Moreover, $\L/p$ being a preorder makes the colimit  $\overline{T}(p)$ of the diagram $\D_T(p) : \L/p \to \GM$ special.
Informally, only maximal instances matter, sub-instances being used to specify how to amalgamate the \rhs of maximal instances.
Formally, whenever we have a morphism $\langle e, f' \rangle : \langle \gamma, f' \circ \L(e) \rangle \to \langle \gamma', f' \rangle \in \L/p$, the \rhs of $\gamma'$ contains the \rhs of $\gamma$ through $e : \gamma \to \gamma'$.
With $f = f' \circ \L(e)$, this means that $\langle \gamma, f \rangle$ does not contribute more data to the output.
The role of the non-maximal $\langle \gamma, f \rangle$ is to specify how the \rhs of some $\gamma''$ should be aligned with the \rhs of $\gamma'$ in the resulting presheaf when there is a second morphism $\langle \gamma, f \rangle \to \langle \gamma'', f'' \rangle$ to another maximal instance $\langle \gamma'', f'' \rangle$.

As an example, consider the diagram depicted on bottom right of Fig.~\ref{fig:sierpinskyStep} and its colimit on the top right.
All the elements of the colimit are given in the \rhs of the maximal instances (the 3 triangles).
The minimal instances (the 6 one-vertex graphs) are used to specify how the 9 vertices of the 3 refined triangles should be merged to get the colimit.

\paragraph*{Computing $\overline{T}$ Online.}

Given a presheaf $p \in \GM$, we call a \emph{partial decomposition of $p$ with respect to $\L$} a subset $M$ of maximal instances of $\L/p$ such that the restriction $\widetilde{M}$ of $\L/p$ to $M$ and morphisms into $M$ remains connected.
We write $\widetilde{\L/p}$ for  the category of partial decompositions of $p$ with set inclusions as morphisms.
$\widetilde{\L/p}$ represents the different ways $\L/p$ can be visited from maximal instance to maximal instance by the use of non-maximal instances to guide the merge.
We extend the action of $\overline{T}$ on $p$ into a function {$\widetilde{T}_p : \widetilde{\L/p} \to \GM$} as follows:
\begin{equation}
    \widetilde{T}_p(M) = \Colim(\D_T(p) \restriction \widetilde{M}).
    \label{eq:tilde}
\end{equation}
The definition $\widetilde{T}_p$ is in fact a complete functor also acting on any morphism $M \subseteq M'$ of $\widetilde{\L/p}$ using the exact same construction as given in Remark~\ref{rem:overlineT} for $\overline{T}$.

In the case of Fig.~\ref{fig:sierpinskyStep}, the maximal instances being the three triangles, the partial decompositions consist of subsets having 0, 1, 2 or 3 of these triangles.
As an example, choose arbitrarily two of these triangles, say $t_1$ and $t_2$.
The intermediate result $\widetilde{T}_p(\{t_1, t_2\})$ is the graph consisting of the two refinements of $t_1$ and $t_2$ glued by their common vertex.

The online computation of $\overline{T}(p)$ consists in iterating a simple step that builds $\widetilde{T}_p(M \cup \{ m \})$ from $\widetilde{T}_p(M)$ as soon as a new maximal instance $m$ has been discovered.
This step is the local amalgamation of $\D_T(p)(m)$ with $\widetilde{T}_p(M)$ considering all the non-maximal sub-instances, say $\{ n', \ldots, n'' \}$, shared by the elements of $M$ and $m$.
Indeed each such sub-instance $n$ gives rise to a span $\widetilde{T}_p(M) \from \D_T(p)(n) \to \D_T(p)(m)$.
Gathering all these spans leads to the \emph{suture} diagram $S_{M,m}$ defined as follows:
\begin{equation}
\label{eq:suture}
\begin{tikzpicture}[xscale=1,yscale=0.4, baseline=(current  bounding  box.center)]
    \node (i1) at (0,0) {$D(n')$};
    \node (dots) at (1,0) {$\dots$};
    \node (i2) at (2,0) {$D(n'')$};
    \node (Ci) at (0,4) {$\widetilde{T}_p(M)$};
    \node (mi) at (3,2) {$D(m)$};
    \draw[->] (i1) to node[above left]{\small $\widetilde{T}_p(M)_{n'}$} (Ci);
    \draw[->] (i2) to node[above,xshift=0.6em,yshift=0.5em]{\small $\widetilde{T}_p(M)_{n''}$} (Ci);
    \draw[->] (i1) to node[right,xshift=-0.4em,yshift=1.1em]{\small $D(e')$} (mi);
    \draw[->] (i2) to node[right,yshift=-.3em]{\small $D(e'')$} (mi);
\end{tikzpicture}
\end{equation}
where $D$ stands for $\D_T(p)$ for simplicity, notation that we will use from now on.
The colimit of a single span being called a pushout, we call the colimit of this collection of spans a \emph{generalized pushout}.
The fact that this generalized pushout $\Colim(S_{M,m})$ and the desired colimit $\widetilde{T}_p(M \cup \{ m \})$ as given in Eq.~\eqref{eq:tilde} do coincide is formalized by the following proposition.
\begin{proposition}
\label{prop:online_step}
Let $M' = M \cup \{ m \} \in \widetilde{\L/p}$. As a cocone, $\widetilde{T}_p(M')$ has the same apex as $\Colim(S_{M,m})$ and has components 
\begin{equation*}
\widetilde{T}_p(M')_n =
\begin{cases}
\Colim(S_{M,m})_{\widetilde{T}_p(M)} \circ \widetilde{T}_p(M)_n & \text{if } n \in \widetilde{M}, \\
\Colim(S_{M,m})_{D(m)} \circ D(e) & \text{for any } e: n \mto m,
\end{cases}
\end{equation*}
where $S_{M,m}$ is diagram~\eqref{eq:suture}.
\end{proposition}
\begin{proof}
    For simplicity, we write $S$ for $S_{M,m}$ and $K$ (resp. $K'$) for $\widetilde{T}_p(M)$ (resp. $\widetilde{T}_p(M')$) as a cocone.
    Let us consider a bigger diagram containing the diagrams $D \restriction \widetilde{M}$, $S$ and $D \restriction \widetilde{M'}$.
    For this, we take $D \restriction \widetilde{M'}$ and add the object $\widetilde{T}_p(M)$ and all morphisms $K(o) : D(o) \to \widetilde{T}_p(M)$ for $o \in \widetilde{M}$.
    By universality of $\widetilde{T}_p(M)$, there is a bijection between cocones on this big diagram and cocones on 
    $D \restriction \widetilde{M'}$.
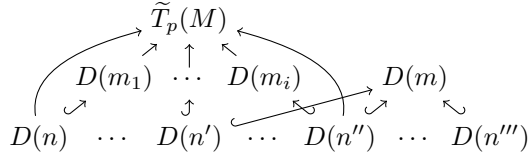
    \begin{center}
        \begin{tikzpicture}[xscale=1,yscale=0.4]
        \node (n1)  at (0,0) {$D(n)$};
        \node (d12) at (1,0) {$\dots$};
        \node (n2)  at (2,0) {$D(n')$};
        \node (d23) at (3,0) {$\dots$};
        \node (n3)  at (4,0) {$D(n'')$};
        \node (d34) at (5,0) {$\dots$};
        \node (n4)  at (6,0) {$D(n''')$};
        \node (m1)   at (1,2) {$D(m_1)$};
        \node (d1i)  at (2,2) {$\dots$};
        \node (mi)   at (3,2) {$D(m_i)$};
        \node (mip1) at (5,2) {$D(m)$};
        \node (gi) at (2,4) {$\widetilde{T}_p(M)$};
        \draw[right hook->] (n1) to (m1);
        \draw[right hook->,shorten >=4pt] (n2) to (d1i);
        \draw[right hook->] (n2) to (mip1);
        \draw[left hook->] (n3) to (mi);
        \draw[right hook->] (n3) to (mip1);
        \draw[left hook->] (n4) to (mip1);
        \draw[->] (m1) to (gi);
        \draw[->] (d1i) to (gi);
        \draw[->] (mi) to (gi);
        \draw[->] (n1) to[bend left] (gi);
        \draw[->] (n3) to[bend right] (gi);
        \end{tikzpicture}
    \end{center}
    where $M = \{ m_1, \ldots, m_i \}$.
    Adding to the picture $\Colim(S)$ with its associated components $t: \widetilde{T}_p(M) \to \Colim(S)$ and $c: D(m) \to \Colim(S)$, recall that $K'_{o} = t \circ K_{o}$ for all $o \in \widetilde{M}$, and $K'_{o} = c \circ D(e)$ for all $e : o \mto m$.
    This is indeed a cocone on $D \restriction \widetilde{M'}$ by properties of $t$ and $c$, and extending it with $K'_{\widetilde{T}_p(M)} = t$ gives us the unique corresponding cocone on the big diagram.
    To establish that it is universal, let us consider another arbitrary cocone $W$ on $D \restriction \widetilde{M'}$.
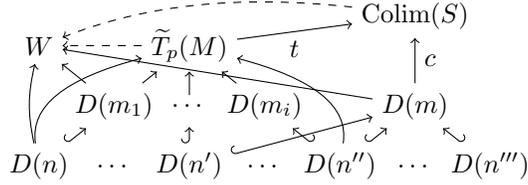
    \begin{center}
        \begin{tikzpicture}[xscale=1,yscale=0.4]
        \node (n1)  at (0,0) {$D(n)$};
        \node (d12) at (1,0) {$\dots$};
        \node (n2)  at (2,0) {$D(n')$};
        \node (d23) at (3,0) {$\dots$};
        \node (n3)  at (4,0) {$D(n'')$};
        \node (d34) at (5,0) {$\dots$};
        \node (n4)  at (6,0) {$D(n''')$};
        \node (m1)   at (1,2) {$D(m_1)$};
        \node (d1i)  at (2,2) {$\dots$};
        \node (mi)   at (3,2) {$D(m_i)$};
        \node (mip1) at (5,2) {$D(m)$};
        \node (gi) at (2,4) {$\widetilde{T}_p(M)$};
        \draw[right hook->] (n1) to (m1);
        \draw[right hook->,shorten >=4pt] (n2) to (d1i);
        \draw[right hook->] (n2) to (mip1);
        \draw[left hook->] (n3) to (mi);
        \draw[right hook->] (n3) to (mip1);
        \draw[left hook->] (n4) to (mip1);
        \draw[->] (m1) to (gi);
        \draw[->] (d1i) to (gi);
        \draw[->] (mi) to (gi);
        \draw[->] (n1) to[bend left] (gi);
        \draw[->] (n3) to[bend right] (gi);
        \node (gip1) at (5,5) {$\Colim(S)$};
        \draw[->] (gi)   to node[below]{$t$} (gip1);
        \draw[->] (mip1) to node[right]{$c$} (gip1);
        \node (w) at (0,4) {$W$};
        \draw[->] (n1) to[bend left=5] (w);
        \draw[->] (m1) to (w);
        \draw[->] (mip1) to (w);
        \draw[dashed,->] (gi) to (w);
        \draw[dashed,->] (gip1) to[bend right] (w);
        \end{tikzpicture}
    \end{center}
    By diagram chasing, one can see that this cocone can be restricted to its components on $\widetilde{M}$, \ie on $\{n,\ldots,n'',m_1,\ldots,m_i\}$.
    This gives us a unique extension of $W$ on all the big diagram.
    By restricting again to its components on $S$, \ie components at $n',\ldots,n'',\widetilde{T}_p(M),m$, and by universality of $\widetilde{T}_p(M')$, we obtain a unique mediating from $\widetilde{T}_p(M')$ to $W$, as wanted.
    \qed
\end{proof}

As an illustration, for computing $\widetilde{T}_p(\{t_1,t_2,t_3\})$ (which is in fact the output in Fig.~\ref{fig:sierpinskyStep}), it is enough to amalgamate $\widetilde{T}_p(\{t_1,t_2\})$ (already considered) with the refinement $D(t_3)$ of the last triangle $t_3$, using as suture the two vertices of $t_3$ shared with $t_1$ and $t_2$ playing the role of $n'$ and $n''$ in diagram~\eqref{eq:suture}.

\begin{remark}
\label{rem:induction}
To summarize, computing $\overline{T}(p)$ online is a matter of collecting the finite set of all maximal instances $\{ m_1, m_2, \ldots, m_k \}$ of $\L/p$ in any order satisfying that $m_{i+1}$ is connected to $\widetilde{M_i}$ where $M_i =\{m_1,\ldots,m_i\}$.
This allows to replace the single colimit computation of the whole diagram, as in Eq~\eqref{eq:gtstep}, by a sequence of smaller colimit computations using the induction relation:
\begin{equation}
    \widetilde{T}_p(M_i) =
    \begin{cases}
        D(m_1) & \text{if } i = 1\\
        \Colim(S_{M_{i-1},m_{i}}) & \text{otherwise.}
    \end{cases}
\end{equation}
The base case is obtained from Eq.~\ref{eq:tilde} for a singleton set of maximal instance, and the inductive one is
established by Prop.~\ref{prop:online_step}, $S_{M_{i-1},m_{i}}$ being a generalized pushout diagram linking $\widetilde{T}_p(M_{i-1})$ and $D(m_i)$.
The final value $\widetilde{T}_p(M_k)$ of this sequence is the colimit of $\widetilde{M_k}$.
$\widetilde{M_k}$ is exactly the diagram $D$ without the arrows between the non-maximal instances.
But the colimit $\widetilde{T}_p(M_k)$ of $\widetilde{M_k}$ is necessarily the same as the colimit $\overline{T}(p)$ of $D$ by the following proposition.

\begin{proposition}
\label{prop:comma_final}
The subcategory $\widetilde{M_k}$ of $\L/p$ given by all instances but only morphisms to maximal instances is final in $\L/p$, in the sense of final functor.
\end{proposition}
\begin{proof}
We need to show that for any instance $o \in \L/p$, and any two morphisms $e_0 : o \mto o_0$ and $e_1 : o \mto o_1$, there is a zig-zag $z$ in the subcategory and a sequence of morphisms of $\L/p$ from $o$ to each $z_k$ that commutes with each $\overline{z}_k$.
Take $i \in \{0,1\}$.
If $o_i$ is maximal we set $\overline{z}_i = e_i$.
If it is non-maximal, we set $\overline{z}_i = e'_i \circ e_i$ for any $e'_i : o_i \mto o'_i$ to some maximal $o'_i$.
We have just built a valid $z$ of length 2 in the subcategory, the associated sequence of morphisms being simply $\overline{z}_0, id_o, \overline{z}_1$, which trivially commutes as wanted.
\qed
\end{proof}
\end{remark}

\paragraph*{Accretive Rule Systems.}

We are interested in those rule systems where the intermediate results stay inside $\GI$, \ie, such that $\widetilde{T}_p(M \subseteq M')$ are monomorphisms for any $p$ and any $M \subseteq M' \in \widetilde{\L/p}$.
This leads to the following definition of accretive rule systems.

\begin{definition}
An \emph{accretive rule system} $T$ is a rule system such that for any $p\in \GM$, $\widetilde{T}_p$ factors through the forgetful functor $\U : \GI \to \GM$.
\end{definition}

\begin{example}\label{ex:exb-accretive}
    The rule system of Fig.~\ref{fig:figexb} is accretive.
    Focusing on connected $\L/p$, its \lhs implies that $p$ is a connected graph.
    Any $M \in \widetilde{\L/p}$ corresponds to a connected sub-graph of $p$ and is sent to the single vertex graph if it is not empty, or to the empty graph otherwise.
    So for any relation $M \subseteq M'$,
    $\widetilde{T}_p(M \subseteq M')$ is the empty morphism or the identity morphism, and both are monomorphisms.
\end{example}
    
\begin{example}\label{ex:exd-accretive}
    The rule system of Fig.~\ref{fig:figexd} is also accretive.
    Again, $p$ is a connected graph and any $M \in \widetilde{\L/p}$ corresponds to a connected sub-graph of $p$.
    Here $M$ is sent to the same graph with parallel edges simplified into single edge.
    So for any relation $M \subseteq M'$, $M'$ is sent either to the same thing as $M$ when $M'$ only adds more parallel edges, or to a strictly greater graph otherwise.
    In both case $\widetilde{T}_p(M \subseteq M')$ is a monomorphism.
\end{example}
    
\begin{example}\label{ex:exa-accretive}
    On the contrary, the rule system of Fig.~\ref{fig:figexa} is not accretive.
    Consider the cycle $c_3$ of length 3, and the associated $\L/c_3$.
    The latter contains 3 instances of rule $\gamma_1$ and 3  maximal instances of the rule $\gamma_2$.
    Consider the relation $M \subseteq M'$ where $M'$ contains all three maximal instances and $M$ only two of them.
    We have $\widetilde{T}_p(M) = p_3$ and $\widetilde{T}_p(M') = c_3$, but there is no monomorphisms between these two graphs.
\end{example}
    
\begin{example}\label{ex:exc-accretive}
    By the exact same reasoning, the rule system of Fig.~\ref{fig:figexc} is also not accretive.
\end{example}

\subsection{Incremental Rule Systems and Global Transformations}
\label{sec:incremental}

We are interested in giving sufficient conditions for rule systems to be global transformations.
These conditions also imply accretiveness.

Our strategy consists in preventing any super-rule to merge by itself the \rhs of its sub-rules.
In other words, the rule only adds new elements to the \rhs of its sub-rules in an \emph{incremental} way.
A positive expression of this constraint is as follows: if the \rhs of two rules overlap in the \rhs of a common super-rule, this overlap must have been required by some common sub-rules.

\begin{definition}\label{def:well-inter}\label{def:incremental}
Given a rule system $T = \langle \Gam, \L, \R \rangle$, we say that a rule $\gamma \in \Gam$ is incremental if for any two sub-rules $\gamma_1 \xrightarrow{i_1} \gamma \xleftarrow{i_2} \gamma_2$ in $\Gam$, any representable presheaf $\y{c}$, and any $\R(\gamma_1) \overset{x_1}{\leftarrow} \y{c} \overset{x_2}{\to} \R(\gamma_2)$ such that $\R(i_1) \circ x_1 = \R(i_2) \circ x_2$, there are some $\gamma_1 \overset{\pi_1}{\from} \gamma' \overset{\pi_2}{\to} \gamma_2$ and $x: \y{c} \to \R(\gamma')$ such that the following diagrams commute.
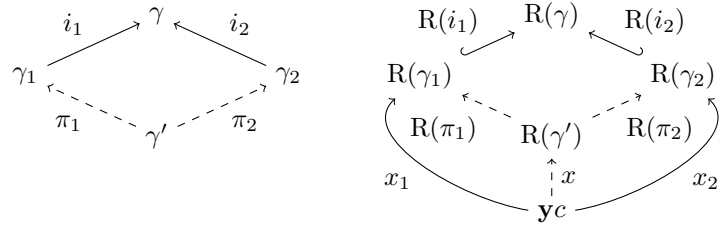
\begin{center}
\begin{tikzpicture}[xscale=1.75,yscale=-.8]
    \begin{scope}[xshift=-3cm]
        \node (g)  at (1,0) {$\gamma$};
        \node (g1) at (0,1) {$\gamma_1$};
        \node (g2) at (2,1) {$\gamma_2$};
        \node (gp) at (1,2) {$\gamma'$};
        \node (a) at (1,3) {};
        \draw[->] (g1) to node[above left]{$i_1$} (g);
        \draw[->] (g2) to node[above right]{$i_2$} (g);
        \draw[->,dashed] (gp) to node[below left]{$\pi_1$} (g1);
        \draw[->,dashed] (gp) to node[below right]{$\pi_2$} (g2);
    \end{scope}
    \begin{scope}
        \node (g)  at (1,0) {$\R(\gamma)$};
        \node (g1) at (0,1) {$\R(\gamma_1)$};
        \node (g2) at (2, 1) {$\R(\gamma_2)$};
        \node (gp) at (1,2) {$\R(\gamma')$};
        \node (a)  at (1,3.3) {$\y{c}$};
        \draw[right hook->] (g1) to node[above left]{$\R(i_1)$} (g);
        \draw[left hook->] (g2) to node[above right]{$\R(i_2)$} (g);
        \draw[->,dashed] (gp) to node[below left]{$\R(\pi_1)$} (g1);
        \draw[->,dashed] (gp) to node[below right]{$\R(\pi_2)$} (g2);
        \draw[->] (a) .. controls +(-0.7,-0.25) and +(-0.5,1) .. node[below left]{$x_1$} (g1);
        \draw[->,dashed] (a) to node[right]{$x$} (gp);
        \draw[->] (a) .. controls +(0.7,-0.25) and +(0.5,1) .. node[below right]{$x_2$} (g2);
    \end{scope}
\end{tikzpicture}
\end{center}
A rule system $T$ is said \emph{incremental} if every $\gamma \in \Gam$ is incremental.
\end{definition}

The Sierpinski gasket rule system (Fig.~\ref{fig:sierpinskyRules}) is incremental.
The only non-trivial case is when the sub-rules $\gamma_1$ and $\gamma_2$ of Def.~\ref{def:incremental} are set to the edge rule of Fig.~\ref{fig:sierpinskyRules} and $\gamma$ to be the triangle rule, such that the \rhs of $\gamma_1$ and $\gamma_2$ overlap on a common vertex in $\R(\gamma)$ (morphisms $x_1$ and $x_2$ of Def.~\ref{def:incremental}) .
This vertex is nothing but the image of the vertex of $\L(\gamma)$ common to the inclusions of $\L(\gamma_1)$ and $\L(\gamma_2)$ in $\L(\gamma)$.
This invites us to set $\gamma'$ to the vertex rule of Fig.~\ref{fig:sierpinskyRules} and complete the requirements of Def.~\ref{def:incremental} to get incrementality.

The main constraint enforced by the incrementality criterion is that any merge is always required by sub-rules as stated by the following lemma.

\begin{lemma}
\label{lemma:zig-zag-reduce-0}
Given an incremental rule system $T = \langle \Gam, \L, \R \rangle$, a zig-zag $z$ in $\Gam$, a representable presheaf $\y{c}$ and two morphisms $x_0: \y{c} \to \R(z_0)$, $x_{|z|}: \y{c} \to \R(z_{|z|})$ such that $z$ links $x_0$ and $x_{|z|}$, \ie there is a cone $\langle x_i: \y{c} \to \R(z_i) \rangle_{0 \leq i \leq |z|}$ that commutes with $\R(z)$.
Then there is a zig-zag $z'$ in $\Gam$ of the form $z_0 = z'_0 \from z'_1 \to z'_2 = z_{|z|}$ that also links $x_0$ and $x_{|z|}$.
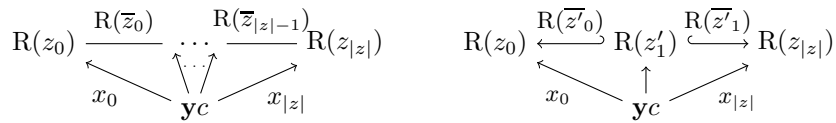
\begin{center}
\begin{tikzpicture}[xscale=1,yscale=0.6]
    \node (z0) at (0,0) {$\R(z_0)$};
    \node (z1a) at (1.7,0) {};
    \node (z1) at (2,0) {$\ldots$};
    \node (z11) at (2,-0.5) {\tiny $\ldots$};
    \node (z1b) at (2.3,0) {};
    \node (zz) at (4,0) {$\R(z_{|z|})$};
    \node (a) at (2,-1.5) {$\y{c}$};
    \draw[-] (z0) to node[above]{\small $\R(\overline{z}_0)$} (z1a);
    \draw[-] (z1b) to node[above]{\small $\R(\overline{z}_{|z|-1})$} (zz);
    \draw[->] (a) to node[below left]{$x_0$} (z0);
    \draw[->] (a) to node[below left]{} (z1a);
    \draw[->] (a) to node[below left]{} (z1b);
    \draw[->] (a) to node[below right]{$x_{|z|}$} (zz);
    
    \node (y0) at (6,0) {$\R(z_0)$};
    \node (y1) at (8,0) {$\R(z'_1)$};
    \node (yz) at (10,0) {$\R(z_{|z|})$};
    \node (aa) at (8,-1.5) {$\y{c}$};
    \draw[left hook->] (y1) to node[above]{\small $\R(\overline{z'}_0)$} (y0);
    \draw[right hook->] (y1) to node[above]{\small $\R(\overline{z'}_1)$} (yz);
    \draw[->] (aa) to node[below left]{\small $x_0$} (y0);
    \draw[->] (aa) to node[below left]{} (y1);
    \draw[->] (aa) to node[below right]{\small $x_{|z|}$} (yz);
\end{tikzpicture}
\end{center}
\end{lemma}

\begin{proof}
This is proved by induction on the length of $z$.
We show the base case with $|z| = 0$ by taking $z' = \langle id_{z_0}, id_{z_0} \rangle$.
For the induction case we assume that the proposition is true for any zig-zag of size $k$ and take $z$ of size $k+1$.
Then we have two cases that depends on the direction of the first morphism of $z$:
\begin{itemize}

\item If $\overline{z}_0: z_1 \to z_0$ we can apply the induction hypothesis on the zig-zag $y = \langle \overline{z}_1, \dots \overline{z}_{k+1}\rangle$ to get a zig-zag $y'$ of the form $z_1 = y'_0 \from y'_1 \to y'_2 = z_{k+1}$ that links $x_1$ and $x_{k+1}$.
Then we take $z'= \langle \overline{z}_0 \circ \overline{y'}_0, \overline{y'}_1 \rangle$ to conclude this case.

\item If $\overline{z}_0: z_0 \to z_1$ we first apply the induction hypothesis on $y = \langle \overline{z}_1, \dots \overline{z}_{k+1}\rangle$ to get a zig-zag $y'$ of the form $z_1 = y'_0 \from y'_1 \to y'_2 = z_{k+1}$ that links $x_1$ and $x_{k+1}$.
Let $x' : \y{c} \to \R(y'_1)$ be the induced linking morphism such that $\R(\overline{z}_0) \circ x_0 = \R(\overline{y'}_0) \circ x'$.
Applying Def.~\ref{def:well-inter} on this square, we get a quadruplet $\langle \gamma' \in \Gam, \pi_1 : \gamma' \to z_0, \pi_2: \gamma' \to y'_1, x : \y{c} \to \R(\gamma') \rangle$ such that $\R(\overline{z}_0) \circ \R(\pi_1) = \R(\overline{y'}_0) \circ \R(\pi_2)$, $x_0 = \R(\pi_1) \circ x$, and $x' = \R(\pi_2) \circ x$.
Here the wanted $z'$ is $\langle \pi_1, \overline{y'}_1 \circ \pi_2 \rangle$.\qed
\end{itemize}
\end{proof}

Consider any monomorphism $h: p \mto p'$ of presheaves such that some merge is required by the computation of $\overline{T}(p')$ between some elements of \rhs instances also involved by $\overline{T}(p)$.
Lemma~\ref{lemma:zig-zag-reduce-0} ensures that it is required by a sub-rule which must also be instantiated by $\overline{T}(p)$ so that the merge is also required by the computation of $\overline{T}(p)$.
In other words, $\overline{T}(h)$ is a monomorphism as established by the following theorem.

\begin{theorem}\label{thm:well-defined-gt-extends}
Any incremental rule system is a global transformation.
\end{theorem}
\begin{proof}
Consider the setting of Remark~\ref{rem:overlineT}.
We need to show that $\overline{T}(f)$ is a monomorphism for $f: p \mto p'$.
For any two morphisms $x_1, x_2: \y{c} \to \overline{T}(p)$ for any $c \in \C$ such that $\overline{T}(f) \circ x_1 = \overline{T}(f) \circ x_2$.
We are left to show that $x_1 = x_2$.

Take $k \in \{1,2\}$.
By Prop.~\ref{prop:colim-graphi}, $x_k$ factors through some cocone component of $\overline{T}(p)$.
Say $x_k = \overline{T}(f)_{\langle \gamma_k, m_k \rangle} \circ y_k$ where $\langle \gamma_k, m_k \rangle$ is a object of $\L/p$.
So $\overline{T}(f) \circ x_k = \overline{T}(f) \circ \overline{T}(p)_{\langle \gamma_k, m_k \rangle} \circ y_k = \overline{T}(p')_{\langle \gamma_k, f \circ m_k \rangle} \circ y_k$.
Since $\overline{T}(f) \circ x_1 = \overline{T}(f) \circ x_2$ then $\overline{T}(p')_{\langle \gamma_1, f \circ m_1 \rangle} \circ y_1 = \overline{T}(p')_{\langle \gamma_2, f \circ m_2 \rangle} \circ y_2$.
Using Prop.~\ref{prop:colim-graphi} on the latter equality into the colimit $\overline{T}(p')$, there is a zig-zag $z$ in $\L/p'$ from $\langle \gamma_1, f \circ m_1 \rangle$ to $\langle \gamma_2, f \circ m_2 \rangle$ that links $y_1$ and $y_2$ through $\D(p')$.
The zig-zag $\Proj[\L/p'](z)$ in $\Gamma$ obviously links $y_1$ and $y_2$ through $\U \circ \R$.
Since the rule system $T$ is incremental, we apply Lemma~\ref{lemma:zig-zag-reduce-0} to obtain a zig-zag $z'$ in $\Gamma$ of the form $\gamma_1 = z'_0 \from z'_1 \to z'_2 = \gamma_2$ that links $y_1$ and $y_2$ through $\U \circ \R$.
Let us define the zig-zag $z''$ in $\L/p$ to be $\langle \gamma_1, m_1 \rangle \mfrom \langle z'_1, h \rangle \mto \langle \gamma_2, m_2 \rangle$ where $h = m_1 \circ \L(\overline{z'}_1) = m_2 \circ \L(\overline{z'}_2)$ and $\overline{z''}_j = \langle \overline{z'}_k, m_k \rangle$ for $k \in \{ 1 , 2\}$.
Clearly, $\Proj[\L/p](z'') = z'$, so $z''$ links $y_1$ and $y_2$ through $\D(p)$.
By commutation properties of cocones over $\D(p)$, we have that $\overline{T}(p)_{\langle \gamma_1, m_1 \rangle} \circ y_1 = \overline{T}(p)_{\langle \gamma_2, m_2 \rangle} \circ y_2$, which implies $x_1 = x_2$ as wanted.
\qed
\end{proof}

The previous remark also applies for intermediate results leading to the following theorem concerning accretiveness.

\begin{theorem}\label{thm:well-defined-accretive}
    Any incremental rule system is accretive.
\end{theorem}
\begin{proof}
    The proof is similar to the proof of Theorem~\ref{thm:well-defined-gt-extends}.
    \qed
\end{proof}

However, the converses of these theorems do not hold so incrementality is sufficient but not necessary as illustrated by the following examples.

\begin{example}
    The rule system of Fig.~\ref{fig:figexc} is a global transformation as explained in Example~\ref{ex:exc-gt}, but not incremental.
    Consider $e_1: \gamma_1 \to \gamma_2$ be the plain arrow into $\gamma_2$ and $e_2: \gamma_1 \to \gamma_2$ the dashed arrow into $\gamma_2$.
    The cospan $\gamma_1 \mto \gamma_2 \mfrom \gamma_1$ is such that $h_1: d_1 \mto \R(\gamma_1)$ and $h_2: d_1 \mto \R(\gamma_1)$ such that $\R(e_1) \circ h_1 = \R(e_2) \circ h_2$ but there is no rule $\gamma'$ to ensure the incrementality condition.
\end{example}

\begin{example}
    Similarly, the rule system of Fig.~\ref{fig:figexd} is a global transformation (Example~\ref{ex:exd-gt}) but is not incremental.
    Consider $e_1: \gamma_2 \to \gamma_3$ be the plain arrow into $\gamma_3$ and $e_2: \gamma_2 \to \gamma_3$ the dashed arrow into $\gamma_3$.
    The cospan $\gamma_2 \mto \gamma_3 \mfrom \gamma_2$ is such that $h_1: l_1 \mto \R(\gamma_2)$ and $h_2: l_1 \mto \R(\gamma_2)$ such that $\R(e_1) \circ h_1 = \R(e_2) \circ h_2$ but there is no rule $\gamma'$ to ensure the incrementality condition.
\end{example}

\begin{example}
    The rule system of Fig.~\ref{fig:figexb} is accretive (Example~\ref{ex:exb-accretive}) but non-incremental. Consider $e_1: \gamma_1 \to \gamma_2$ the plain arrow into $\gamma_3$ and $e_2: \gamma_1 \to \gamma_2$ the dashed arrow into $\gamma_3$.
    Observe that for the cospan $\gamma_1 \mto \gamma_2 \mfrom \gamma_1$ we have $h_1: d_1 \mto \R(\gamma_2)$ and $h_2: d_1 \mto \R(\gamma_2)$ such that $\R(e_1) \circ h_1 = \R(e_2) \circ h_2$ but there is no rule $\gamma'$ to ensure the incremental condition.
\end{example}

\begin{example}
    Similarly, the rule system of Fig.~\ref{fig:figexd} is accretive (Example~\ref{ex:exd-accretive}) but non-incremental.
\end{example}

Summarizing the properties collected with the four examples of Fig.~\ref{fig:figex} and with the one of Fig.~\ref{fig:sierpinskyRules} in a table, we can see that being a global transformation and being accretive are orthogonal properties, but incrementality forces the two.

\begin{center}
\begin{tabular}{|r|c|c|c|c|}
    \cline{2-5}
    \multicolumn{1}{c|}{}
    &
    \multicolumn{2}{c|}{non-incr.} &
    \multicolumn{2}{c|}{incr.} \\
    \cline{2-5}
    \multicolumn{1}{c|}{}
    &
    non-G.T. & G.T. &
    non-G.T. & G.T. \\
    \hline
    {\ }non-accretive{\ } &
    {\ }ex. Fig.~\ref{fig:figexa}{\ } & {\ }ex. Fig.~\ref{fig:figexc}{\ } & {\ }None, Thm.~\ref{thm:well-defined-gt-extends}/\ref{thm:well-defined-accretive} {\ } & {\ } None, Thm.~\ref{thm:well-defined-accretive} {\ } \\
    \hline
    accretive{\ } &
    {\ }ex. Fig.~\ref{fig:figexb}{\ } & {\ }ex. Fig.~\ref{fig:figexd} {\ }&{\ } None, Thm.~\ref{thm:well-defined-gt-extends} {\ }&{\ } Sierpenski {\ }\\
    \hline
\end{tabular}
\end{center}

\section{Computing Accretive Global Transformations}
\label{sec:algo}

This section is devoted to the description of an effective implementation of the online procedure considered in Section~\ref{sec:accretive}.
In this context, we focus on incremental global transformations.
We first explain how the categorical concepts of Section~\ref{sec:online} are represented computationally (Section~\ref{sec:rep}).
This is followed by a detailed presentation of the algorithm (Section~\ref{sec:actual-algo}).

\subsection{Categorical Constructions Computationally}
\label{sec:rep}

Up to now, we exposed everything formally using categorical concepts:
the category of presheaves $\GI$,
finite incremental rule systems $T = \langle \Gam, \L, \R \rangle$, and
the comma category $\L/p$ for some finite presheaf $p$.
We now describe their computational counterparts.
First, let us introduce some notations used in the algorithm:
\begin{itemize}

\item $X^*$ stands for the set of finite words on the alphabet $X$; the empty word is denoted by $\varepsilon$ and the concatenation by $u \cdot v$ for any two words $u, v \in X^*$.

\item $\coprod_{a \in A} B(a)$ is the set of pairs $(a, b)$ where $a \in A$ and $b \in B(a)$.

\item $\prod_{a \in A} B(a)$ is the set of functions $f: A \to \bigcup_{a \in A} B(a)$ such that for any $a \in A$, $f(a) \in B(a)$. Such functions are also manipulated as sets of pairs. Those pairs are written $a \mapsto f(a)$.

\end{itemize}

\paragraph*{The Category of Presheaves with Monomorphisms.}

The category $\GI$ is the formal abstraction for a library providing a data structure suitably captured by presheaves (like sets, graphs, Petri nets, etc.) and how an instance of that data structure (presheaves) is part of another one (monomorphisms).
Two functions $\-- \circ \--$ and $\-- = \-- $ need to be provided to compute composition and equality test of sub-parts.
The library also needs to come with a pattern matching procedure taking as input two finite presheaves $p$ and $p'$ and returning the set $\Hom{\GI}{p}{p'}$ of occurrences of $p$ in $p'$.
Finally, the library is assumed to provide a particular construction operation called $\mathtt{generalizedPushout}(p_1,p_2,S)$ computing the generalized pushout, \ie, the colimit of the collection of spans $S$, each span being represented as a triplet $(p \in \GI, f_1 : p \mto p_1, f_2 : p \mto p_2)$.
The resulting colimit is returned as a triplet $(c \in \GI, g_1 : p_1 \mto c, g_2 : p_2 \mto c)$ where $c$ is the apex and $g_1, g_2$ the corresponding component morphisms.

\paragraph*{Finite Incremental Rule System.}

A finite rule system is described as a finite graph whose vertices are rules $l \Rightarrow r$ as pairs of presheaves and edges are pairs of monomorphisms $\langle i_l : l_1 \mto l_2, i_r : r_1 \mto r_2 \rangle$.
Functors $\L$ and $\R$ return the first and second components of these pairs respectively.
At the semantic level, $\Gam$ is the category generated from this graph.
Finally, incrementality as presented in Def.~\ref{def:incremental} is clearly decidable on finite rule systems, giving rise to an accretive global transformation by Theorems~\ref{thm:well-defined-gt-extends} and~\ref{thm:well-defined-accretive}.

\paragraph*{The Category of Instances.}

By Prop.~\ref{prop:comma_thin}, $\L/p$ is a preordered set, but in our implementation,
any time an instance is matched, all of its isomorphic instances are taken care of
at the same time.
This corresponds informally to taking the poset of equivalence classes of the preordered set.
Also, by Prop.~\ref{prop:comma_final}, morphisms between non-maximal instances can be ignored.
All in all, $\L/p$ is adequately thought of as an abstract undirected bipartite graph that we call \emph{the network}.

Finally, the $\L/p$ is never entirely represented in memory (neither is the cocone associated to the resulting colimit).
A first instance is constructed, and the others are built from neighbor to neighbor through the operation $\coprod_{n}\Hom{\L/p}{n}{m}$ and $\coprod_{m}\Hom{\L/p}{n}{m}$.
The former lists the sub-instances of $m$ and the latter lists the super-instances of $n$.
For the  ``incoming neighbors'' or sub-instances $\coprod_{n}\Hom{\L/p}{n}{\langle \gamma', f' \rangle}$, they are specified as
$$
\{ (\langle \gamma, f' \circ \L(e) \rangle, \langle e, f' \rangle)
\mid e : \gamma \to \gamma' \}.
$$
This corresponds simply to the composition $\-- \circ \--$ in $\GI$ discussed earlier.
On the contrary, ``outgoing neighbors'' or super-instances $\coprod_{m}\Hom{\L/p}{\langle \gamma', f' \rangle}{m}$ correspond to extensions and are obtained by pattern matching.
\begin{eqnarray*}
    \{ \langle e', f'' \rangle : \langle \gamma', f' \rangle \to \langle \gamma'', f'' \rangle
    \mid \\ e' : \gamma' \to \gamma'', && f'' \in \Hom{\GI}{\L(\gamma'')}{p}%
    \text{ s.t. } f' = f'' \circ \L(e') \}.    
\end{eqnarray*}
Notice that these two specifications are obtained by simply unfolding the definition of the morphisms of the comma category.
Also, the set of incoming morphisms $e : \gamma \to \gamma'$ and outgoing morphisms $e' : \gamma' \to \gamma''$ in $\Gam$ are directly available in the graph representation of the rule system $T$ as said earlier.
These operations are used to implement a breadth-first algorithm, earlier instances being dropped away as soon as their maximal super-instances have been found.

\subsection{The Global Transformation Algorithm}
\label{sec:actual-algo}

\begin{algorithm}[t]
    \SetEndCharOfAlgoLine{}
    \SetKw{Let}{let}
    \SetKwInput{KwInput}{Input}
    \SetKwInput{KwVar}{Variable}
    \KwInput{$T : \text{rule system on \GI}$}
    \KwInput{$p : \GI$}
    \KwVar{$P : \GI$}
    \KwVar{$N : (\L/p)^*$}
    \KwVar{$E \subseteq \coprod_{n \in N}\coprod_{m}\Hom{\L/p}{n}{m}$}
    \KwVar{$C : \prod_{n \in N}\Hom{\GI}{D(n)}{P}$}
    \Let $n = \mathtt{findAnyMinimal}(T,g)$, \ie, any minimal element in $\L/p$\;
    \Let $E = \emptyset$, $C = \{ n \mapsto id_{D(n)}\}$, $N = n$, $P = D(n)$\;
    \While{$N \neq \varepsilon$}{
        \Let $n = \mathtt{head}(N)$, \ie, the first instance in the queue without modifying $N$\;
        \Let $M' = \coprod_{m}\Hom{\L/p}{n}{m}$\;
        \For{$(m,e) \in M'$ s.t. $(n,m,e) \not\in E$ and $m$ is maximal}{
            \Let $E' = \coprod_{n'\neq m}\Hom{\L/p}{n'}{m}$\;
            \Let $S = \{ (n',C(n'),D(e')) \mid (n',e') \in E', n' \in N\}$\;
            \Let $(P',t,c) = \mathtt{generalizedPushout}(P,D(m),S)$\;
            $E := E \cup \{ (n',m,e') \mid (n',e') \in E' \}$\;
            $C := C \cup \{ n' \mapsto c \circ D(e') \mid (n',e') \in E', n' \not\in N  \}$\;
            $N := N \cdot \langle n' \mid (n',e') \in E', n' \not\in N\rangle$\;
            $P := P'$\;
        }
        $E := \{ (n', m, e') \in E \mid n' \neq n \}$\;
        $C := \{ n' \mapsto C(n') \in C \mid n' \neq n \}$\;
        $N := \mathtt{tail}(N)$, \ie, removes the first instance $n$ from the queue\;
    }
    \Return{P}\;
\caption{}
\label{algo:tg}
\end{algorithm}

Algorithm~\ref{algo:tg} gives a complete description of a procedure to compute $T(p)$ online.
The algorithm manages four variables $P$, $N$, $E$ and $C$.
Variable $P$ contains intermediate results and finally the output presheaf.
The part of the network that is kept in memory is represented by variables $N$ and $E$:
$N$ is a queue containing, in order of discovery, the non-maximal instances that might still have a role to play.
$E$ associates each instance in $N$ to the set of its maximal super-instances that have already been processed.
For simplicity, $E$ is not represented as a function from $N$ to sets but as a relation.
The \rhs $D(n)$ of each instance $n \in N$ is already in the current result $P$ through the morphism kept as $C(n)$.

\begin{figure}[t]
    \begin{center}
    \begin{tikzpicture}[scale=0.85]
        \node (D1) at (0, 0) {$D(n_1)$};
        \node (G1) at (3, 0) {$P_1$};
        \node (G2) at (6, 0) {$P_2$};
        \node (G3) at (9.5, 0) {$P_3$};
        \node (G3bis) at (12, 0) {$P_3$};        
        
        \node (D11) at (1, 2) {$D(n_1)$};
        \node (D12) at (4.3, 2) {$D(n_1)$};
        \node (D4)  at (4.5, 1) {$D(n_4)$};
        \node (D13) at (7.5, 2) {$D(n_1)$};
        \node (D2)  at (7.7, 1.333) {$D(n_2)$};
        \node (D5)  at (7.9, 0.667) {$D(n_5)$};
        
        \node (Dd1) at (3, 2) {$D(m_1)$};
        \node (Dd2) at (6, 2) {$D(m_2)$};
        \node (Dd3) at (9.5, 2) {$D(m_3)$};
        
        \path (Dd1) edge[->] node[left]{$c_1$} (G1);
        \path (D11) edge[->] node[above,font=\scriptsize]{$D(e')$} (Dd1);
        \path (D11) edge[->] node[left,font=\scriptsize]{$C(n_1)$} (D1);
        \path (D1)  edge[->] node[below]{$t_1$} (G1);
        
        \path (Dd2) edge[->] node[left]{$c_2$} (G2);
        \path (D12) edge[->] (Dd2);
        \path (D12) edge[->] (G1);
        \path (D4)  edge[->] (Dd2);
        \path (D4)  edge[->] (G1);
        \path (G1)  edge[->] node[below]{$t_2$} (G2);
    
        \path (Dd3) edge[->] node[left]{$c_3$} (G3);
        \path (D13) edge[->] (Dd3);
        \path (D13) edge[->] (G2);
        \path (D2)  edge[->] (Dd3);
        \path (D2)  edge[->] (G2);
        \path (D5)  edge[->] (Dd3);
        \path (D5)  edge[->] (G2);
        \path (G2)  edge[->] node[below]{$t_3$} (G3);
        \path (G3)  edge[double, double distance=.5mm] (G3bis);
        
        \node[inner sep=0.2em] (s1) at (0, 4) [draw] {$n_1$};
        
        \path (1, 2.5) edge (1, 5.5);
        \node[inner sep=0.2em] (s1) at (2.5, 4) [draw] {$n_1$};
        \node[inner sep=0.2em] (s2) at (2.5, 5) [draw] {$n_2$};
        \node[inner sep=0.2em] (s3) at (3.5, 5) [draw] {$n_3$};
        \node[inner sep=0.2em] (s4) at (3.5, 4) [draw] {$n_4$};
        \node (c1) at (3, 4.5) [circle, fill] {};
        \path (s1) edge (c1);
        \path (s2) edge (c1);
        \path (s3) edge (c1);
        \path (s4) edge (c1);
        \path (s4) edge[draw=none] node{$m_1$} (s3);
        
        \path (4.3, 2.5) edge (4.3, 5.5);
        
        \node[inner sep=0.2em] (s1) at (5.5, 4) [draw] {$n_1$};
        \node[inner sep=0.2em] (s2) at (5.5, 5) [draw] {$n_2$};
        \node[inner sep=0.2em] (s3) at (6.5, 5) [draw] {$n_3$};
        \node[inner sep=0.2em] (s4) at (6.5, 4) [draw] {$n_4$};
        \node[inner sep=0.2em] (s5) at (5.5, 3) [draw] {$n_5$};
        \node[inner sep=0.2em] (s6) at (6.5, 3) [draw] {$n_6$};
        \node (c1) at (6, 4.5) [circle, fill] {};
        \node (c2) at (6, 3.5) [circle, fill] {};
        \path (s1) edge (c1);
        \path (s2) edge (c1);
        \path (s3) edge (c1);
        \path (s4) edge (c1);
        \path (s1) edge (c2);
        \path (s4) edge (c2);
        \path (s5) edge (c2);
        \path (s6) edge (c2);
        \path (s4) edge[draw=none] node{$m_1$} (s3);
        \path (s6) edge[draw=none] node{$m_2$} (s4);
        
        \path (7.5, 2.5) edge (7.5, 5.5);
        
        \node[inner sep=0.2em] (s1) at (9.1, 4) [draw] {$n_1$};
        \node[inner sep=0.2em] (s2) at (9, 5) [draw] {$n_2$};
        \node[inner sep=0.2em] (s3) at (10, 5) [draw] {$n_3$};
        \node[inner sep=0.2em] (s4) at (10, 4) [draw] {$n_4$};
        \node[inner sep=0.2em] (s5) at (9, 3) [draw] {$n_5$};
        \node[inner sep=0.2em] (s6) at (10, 3) [draw] {$n_6$};
        \node[inner sep=0.2em] (s7) at (8, 3.5) [draw] {$n_7$};
        \node[inner sep=0.2em] (s8) at (8, 4.5) [draw] {$n_8$};
        \node (c1) at (9.5, 4.5) [circle, fill] {};
        \node (c2) at (9.5, 3.5) [circle, fill] {};
        \node (c3) at (8.5, 4) [circle, fill] {};
        \path (s1) edge (c1);
        \path (s2) edge (c1);
        \path (s3) edge (c1);
        \path (s4) edge (c1);
        \path (s1) edge (c2);
        \path (s4) edge (c2);
        \path (s5) edge (c2);
        \path (s6) edge (c2);
        \path (s1) edge (c3);
        \path (s7) edge (c3);
        \path (s8) edge (c3);
        \path (s2) edge (c3);
        \path (s5) edge (c3);
        \path (s4) edge[draw=none] node{$m_1$} (s3);
        \path (s6) edge[draw=none] node{$m_2$} (s4);
        \path (s7) edge[draw=none] node{$m_3$} (s8);
        
        \path (10.5, 2.5) edge (10.5, 5.5);
        \node[inner sep=0.2em] (s2) at (12, 5) [draw] {$n_2$};
        \node[inner sep=0.2em] (s3) at (13, 5) [draw] {$n_3$};
        \node[inner sep=0.2em] (s4) at (13, 4) [draw] {$n_4$};
        \node[inner sep=0.2em] (s5) at (12, 3) [draw] {$n_5$};
        \node[inner sep=0.2em] (s6) at (13, 3) [draw] {$n_6$};
        \node[inner sep=0.2em] (s7) at (11, 3.5) [draw] {$n_7$};
        \node[inner sep=0.2em] (s8) at (11, 4.5) [draw] {$n_8$};
        \node (c1) at (12.5, 4.5) [circle, fill] {};
        \node (c2) at (12.5, 3.5) [circle, fill] {};
        \node (c3) at (11.5, 4) [circle, fill] {};
        \path (s2) edge (c1);
        \path (s3) edge (c1);
        \path (s4) edge (c1);
        \path (s4) edge (c2);
        \path (s5) edge (c2);
        \path (s6) edge (c2);
        \path (s7) edge (c3);
        \path (s8) edge (c3);
        \path (s2) edge (c3);
        \path (s5) edge (c3);
        \path (s4) edge[draw=none] node{$m_1$} (s3);
        \path (s6) edge[draw=none] node{$m_2$} (s4);
        \path (s7) edge[draw=none] node{$m_3$} (s8);
    
    \end{tikzpicture}
    \end{center}
    \caption{Evolution of the data during the four firsts steps of the algorithm.
    From left to right: 
    we start with a non-maximal instance, process its associated maximal instances successively, and finally drop the non-maximal.
    At each stage, the output is updated by generalized pushout.
    }
    \label{fig:network}
\end{figure}
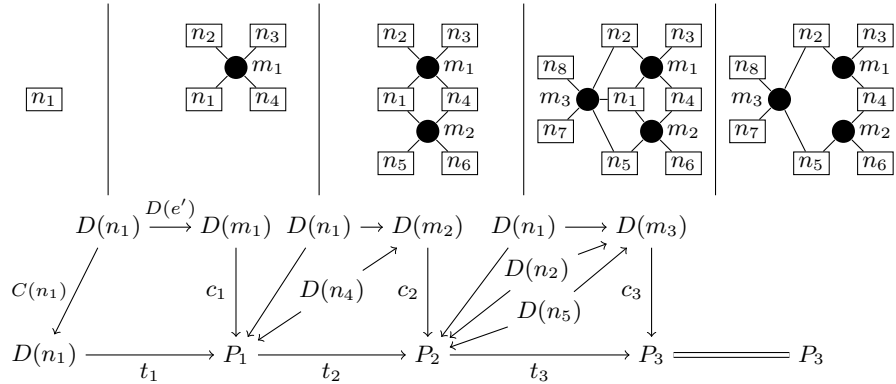

Fig.~\ref{fig:network} illustrates the first steps of Algorithm~\ref{algo:tg} representing maximal instances as black dots, and non-maximal instances as white squares.
The initialization step is to find a first instance (line 1).
For that, we try each minimal pattern, and start with the first founded minimal instance, say $n_1$.
At this point, the first intermediate result $P_0$ is simply the \rhs $D(n_1)$; we memorize the (identity) relationship between $D(n_1)$ and $P_0$, call it $C(n_1): D(n_1) \to P_0$, and enqueue $n_1$ (line 2).
Enqueued non-maximals are treated one after the other (lines 3, 4, 16).
For each, we consider all maximal super-instances of $n_1$ (lines 5, 6).
In Fig.~\ref{fig:network}, we assume three such super-instances $m_1$, $m_2$ and $m_3$.
They are processed one after the other (line 6).

The first iteration processes $m_1$ by taking all its sub-instances $n_1, \ldots, n_4$ (line 7).
The suture $S$ is computed (line 8) by considering all already computed non-maximals (\ie, in $N$) among these sub-instances.
Here, only $n_1$ is already known and serves to define a one-span suture with morphisms $C(n_1) : D(n_1) \to P_0$ and $D(e') : D(n_1) \to D(m_1)$, where $e'$ is the morphism from $n_1$ to $m_1$.
The generalized pushout of $P_0$ and the \rhs $D(m_1)$ is therefore computed and gives the new intermediate result $P_1$ (lines 9, 13).
Since $P_1$ includes the \rhs of all discovered non-maximals $N = \{ n_1, \ldots, n_4 \}$, we memorize as $C(n) : D(n) \to P_1$ for $n \in N$ the locations of these \rhs in $P_1$ (line 11).
The newly discovered non-maximal instances $n_2$, $n_3$ and $n_4$ are enqueued (line 12).

The second iteration processes $m_2$ similarly and all its sub-instances $n_1$, $n_4$, $n_5$, and $n_6$ are computed.
This time, $n_1$ and $n_4$ are used for computing the new intermediate result $P_2$ by generalized pushout using the two spans $\langle n_1, C(n_1) : D(n_1) \to P_1, D(j) : D(n_1) \to D(m_2) \rangle$ and $\langle n_4, C(n_4) : D(n_4) \to P_1, D(k) : D(n_4) \to D(m_2) \rangle$ as suture.
The set $N$ of discovered non-maximals is updated by adding $n_5$ and $n_6$ as well as the locations $C$ of their \rhs in $P_2$.

The processing of $m_3$ is similar and shows no novelty.
At this point non-maximal $n_1$ does not have any further role to play: the \rhs of all its associated maximals are already amalgamated to the current intermediate result.
$n_1$ is dropped together with all data associated to it (lines 14--16), as shown in the last step of Fig.\ref{fig:network}.
Non-maximal instances being processed in the order of first discovery, the next one is $n_2$ in the example.

During these processings, other non-maximal instances see some of their associated maximals being processed.
We have to keep track of this to avoid double processing of maximals which would cause infinite loops (condition at line 6).
This is the role of $E$ to maintain this information.
Clearly $E$ contains only the useful part of the network: edges from maximals to their sub-instances are registered when discovered (line 10) but cleared up as soon as a non-maximal is dropped (line 14).
Considering that non-maximal instances are treated in order of appearance, the algorithm will process the maximals at distance 1 from $n_1$ first, then those at distance $2$, and so on, until the complete connected component of the network is processed.
In memory, there are never stored more than four ``radius'' of instances $d$, $d+1$, $d+2$ and $d+3$ from $n_1$.


\begin{theorem}
Algorithm~\ref{algo:tg} is correct, \ie, the final value of $P$ is $\overline{T}(p)$.
\end{theorem}
\begin{proof}
We ignore the case when $\L/p$ is composed of a single instance, since the algorithm behaves trivially in that case.

Ignoring lines 8, 9, 11, 13, and 15, variables $P$ and $C$, and looking only at non-maximal instances (variable $n$), the algorithm behaves like a usual breadth-first search.
Indeed, the search begins by enqueueing a first non-maximal instance at line 2.
Each iteration of the while loop (line 3) processes the next non-maximal instance $n$ in the queue (line 4), lists all its ``neighbors via a maximal instance'' (lines 5--7) and enqueues those that have not yet been visited (lines 10, 12) before popping $n$ out of the queue (line 16).
Variables $E$ and $N$ serve as the set of visited non-maximal instances.
The reason line 14 can remove all occurrences of $n$ in the set $E$ without creating an infinite loop is that $E$ memorizes the maximal instances $m'$ from which each enqueued non-maximal instance $n'$ has been reached (line 10).
The constraint $(n,m,e) \not\in E$ of the for loop (line 6) prevents this path to be taken in the other direction.

Since all non-maximal instances are assigned to $n$ at line 4, and each maximal instance is a super-instance of some non-maximal instance, we have that $m$ goes through all maximal instances as well (line 6).
Let us call $m_1$, \ldots, $m_k$, the successive values taken by $m$ and define the sequence of set of maximal instances $M_i = \{ m_1, \ldots, m_i\}$ for $i \in \{1,\ldots,k\}$.
The breath-first traversal ensures that each newly considered $m_{i+1}$ is connected to some maximal instance in $M_i$ by some non-maximal sub-instances.
Let us show now that the successive values taken by $P$ at line 13, numbered $P_1$, $P_2$, \ldots, $P_k$, are such that $P_i = \widetilde{T}_p(M_i)$.
Using Remark~\ref{rem:overlineT}, it is enough to show that $P_1 = D(m_1)$ and $P_{i+1} = \Colim(S_{M_{i},m_{i+1}})$.

For $P_1$, consider the first steps of the algorithm before the first execution of line 13.
Call $n_1$ the value of $n$ at line 1 and note that $P$ is assigned to $D(n_1)$ at line 2 and $C(n_1)$ to $id_{D(n_1)} : D(n_1) \to D(n_1)$.
At lines 4, 5 and 6, $n$ is assigned to $n_1$, $m$ to $m_1$ and $e$ to the corresponding morphism from $n_1$ to $m_1$.
Lines 7 and 8 lead $S$ to be $\{ (n_1, id_{D(n_1)} : D(n_1) \to D(n_1), D(e) : D(n_1) \to D(m_1)) \}$.
So the first execution of line 9 computes this simple pushout and sets $(P', t, c)$ to $(D(m_1), D(e), id_{D(m_1)})$, so $P_1 = D(m_1) = \widetilde{T}_p(M_1)$ at line 13.

To establish that $P_{i+1} = \Colim(S_{M_{i},m_{i+1}})$ in the $(i+1)$-th execution of line 13, we need to show that the parameters $(P,D(m),S)$ provided in $(i+1)$-th execution of line 9 correspond to the diagram $S_{M_{i},m_{i+1}}$ given in \eqref{eq:suture}.
Firstly, by induction hypothesis, we have that $P = P_i = \widetilde{T}_p(M_i)$ and $m = m_{i+1}$.
The collection of spans $S$ computed at line 8 is correct because $E'$ is the set of sub-instances of $m$ (line 7), and $N$ contains all sub-instances of the maximal instances in $M_i$ that could have a morphism to $m_{i+1}$.
Indeed, a non-maximal instance is discarded from $N$, $E$ and $C$ (lines 14--16) only after that all of its maximal super-instances have been processed (for loop at lines 5 to 13).
Line 11 and 15 ensure that $C$ always contain the correct morphism $D(n) \to P$ for all non-maximal instances $n$ contained in $N$.

Finally, line 11 modifies $C$ without updating the cocone compounds already stored in $C$, resulting in mixing morphisms with codomain $P$ and $P'$.
It is correct considering the following fact.
For accretive global transformations, $t$ (line 9) is always a monomorphism and can be designed for $t$ to be a trivial inclusion.
In that case, any morphism to $P$ is also a morphism to $P'$, the latter materially including $P$.
In other words, everything is implemented to ensure that the modification on intermediate results are realized \emph{in place}.
\qed
\end{proof}

\let\oldnl\nl
\newcommand{\nonl}{\renewcommand{\nl}{\let\nl\oldnl}}

\section{Conclusion}\label{sec:conclusion}

In this paper, we have presented an online algorithm for computing the application of global transformations on presheaves.
Note that this work was originally restricted to global transformations of graphs but the extension to any category of presheaves appears to be straightforward.
It is natural to expect the extensions to other well-known classes of categories, in particular for the class of ($\mathcal{M}$-)adhesive categories~\cite{lack2005adhesive,ehrig2010categorical}.

At the algorithmic level, there remain many interesting considerations that need to be settled.
One of them is that the way this algorithm goes from maximal instances to maximal instances using common sub-instances reminds of the strategy of the famous Knuth-Morris-Pratt algorithm~\cite{knuth1977fast}: in both cases the content of one match is used to guide following subsequent pattern matching.
This link is reinforced by the work of~\cite{srinivas1993sheaf} that extend the Knuth-Morris-Pratt algorithm to sheaves.
In Algorithm~\ref{algo:tg}, we used pattern matching as a black-box but opening it should allow to mix the outer maximal-to-maximal strategy with the Knuth-Morris-Pratt considerations inside the pattern matching algorithm of~\cite{srinivas1993sheaf}.
Another important aspect is the complexity of this online approach and its natural extensions.
Indeed, we described how a full input is decomposed in an online fashion, and the parts also treated online.
The full picture includes the input itself being received by part, or even treated in a distributed way.
Each of these versions deserve a careful study of their online complexity, \ie, the complexity of the computation happening between each outputted data.
We are also interested in the detailed study of the problem consisting of deciding, given a rule system, if it is a global transformation or not.
Incremental rule systems form a particularly easy sub-class for this problem but we are talking here about the complete class of all rule systems.

The incremental criterion can be studied for itself.
An alternative equivalent expression of Definition~\ref{def:incremental} is stated as follows: given a super-rule, its \rhs contains the \rhs of its sub-rules as if they were considered independently.
Intuitively, this prevents from non-local behavior like collapsing non-empty graphs to a single vertex since the empty graph remains empty for example as in Fig.~\ref{fig:figexb}.
From that point of view, incremental global transformations follow the research direction of causal graph dynamics~\cite{arrighi2018cellular}.
In this work any produced element in the output is attached to an element of the input graph and a particular attention is put on preventing two rule instances to produce a common element.

\bibliographystyle{splncs04-v2}
\bibliography{bib-v2}

\end{document}

%% file: sierpinsky-1-v2.tex
\input{sierpinsky-common-v2.tex}

\begin{tikzpicture}[scale=0.5]
    \begin{scope}[yscale=2/6*1.73]
    \tri{-1}{0}{0}{0.5}{1};
    \draw [color=gray, rounded corners=5*0.66] (-1.5,-0.5) -- (3.5,-0.5) -- (1,7) --cycle;

    \edge{-1}{-4}{0.5}{1};
    \draw [color=gray, rounded corners=5*0.66] (-1.25,-3.5) -- (3.25,-3.5) -- (3.25,-4.5) -- (-1.25,-4.5) --cycle;

    \nod{1}{-8}{0.5}{1};
    \draw [color=gray] (n0) circle (1/3.5 and 1.73/3.5);
    
    \path let \p{e0}=(e0), \p{e1}=(e1), \p{t0}=(t0), \p{t1}=(t1) in
    (\x{e0}/2+\x{e1}/2, \y{e0}) edge[>=stealth, color=blue,thick, shorten >=10*0.75,shorten <=10*0.75, ->] (\x{t0}/2+\x{t1}/2, \y{t0});
    
    \path let \p{e0}=(e0), \p{e1}=(e1), \p{t0}=(t1), \p{t1}=(t2) in
    (\x{e0}/2+\x{e1}/2, \y{e0}/2+\y{e1}/2) edge[>=stealth, color=blue, densely dashed, thick, shorten >=10*0.75,shorten <=20,out=35,in=50,looseness=1.4, ->] (\x{t0}/2+\x{t1}/2, \y{t0}/2+\y{t1}/2);
    
    \path let \p{e0}=(e0), \p{e1}=(e1), \p{t0}=(t2), \p{t1}=(t0) in
    (\x{e0}/2+\x{e1}/2, \y{e0}/2+\y{e1}/2) edge[>=stealth, color=blue, densely dotted, thick, shorten >=10*0.75,shorten <=20, out=145,in=130,looseness=1.4, ->] (\x{t0}/2+\x{t1}/2, \y{t0}/2+\y{t1}/2);
    \node at (-1.5, 5) {$\L(i_3)$};

    \draw [>=stealth, color=red, thick, shorten >=10*0.75,shorten <=10*0.75, ->] (n0) edge[bend left] (e0);
    \draw [>=stealth, color=red, densely dashed, thick, shorten >=10*0.75,shorten <=10*0.75, ->] (n0) edge[bend right] (e1);
    
    \ttri{8}{0}{0.5}{1};
    \draw [>=stealth, color=gray, rounded corners=5*0.66] (7.5,-0.5) -- (12.5,-0.5) -- (10,7) --cycle;

    \eedge{8}{-4}{0.5}{1};
    \draw [>=stealth, color=gray, rounded corners=5*0.66] (7.75,-3.5) -- (12.25,-3.5) -- (12.25,-4.5) -- (7.75,-4.5) --cycle;

    \nod{10}{-8}{0.5}{1};
    \draw [>=stealth, color=gray] (n0) circle (1/3.5 and 1.73/3.5);

    \path let \p{e0}=(e0), \p{e1}=(e1), \p{t0}=(t0), \p{t1}=(t1) in
    (\x{e0}/2+\x{e1}/2, \y{e0}) edge[>=stealth, color=blue,thick, shorten >=10*0.75,shorten <=10*0.75, ->] (\x{t0}/2+\x{t1}/2, \y{t0});
    
    \path let \p{e0}=(e0), \p{e1}=(e1), \p{t0}=(t1), \p{t1}=(t2) in
    (\x{e0}/2+\x{e1}/2, \y{e0}/2+\y{e1}/2) edge[>=stealth, color=blue, densely dashed, thick, shorten >=10*0.75,shorten <=20,out=35,in=50,looseness=1.4, ->] (\x{t0}/2+\x{t1}/2, \y{t0}/2+\y{t1}/2);
    
    \path let \p{e0}=(e0), \p{e1}=(e1), \p{t0}=(t2), \p{t1}=(t0) in
    (\x{e0}/2+\x{e1}/2, \y{e0}/2+\y{e1}/2) edge[>=stealth, color=blue, densely dotted, thick, shorten >=10*0.75,shorten <=20, out=145,in=130,looseness=1.4, ->] (\x{t0}/2+\x{t1}/2, \y{t0}/2+\y{t1}/2);
    \node at (7.5, 5) {$\R(i_3)$};

    \draw [>=stealth, color=red,shorten >=10*0.75,shorten <=10*0.75, ->] (n0) edge[thick, bend left] (e0);
    \draw [>=stealth, color=red, densely dashed, shorten >=10*0.75,shorten <=10*0.75, ->] (n0) edge[thick, bend right] (e1);

    \node (lt) at (4.9, 1.67) {};
    \node (rt) at (6.1, 1.67) {};
    \path (lt) edge[double, ->]node[above,yshift=0.1em]{$\gamma_3$} (rt);
    
    \node (le) at (4.9, -4) {};
    \node (re) at (6.1, -4) {};
    \path (le) edge[double, ->]node[above,yshift=0.1em]{$\gamma_2$} (re);

    \node (ln) at (4.9, -8) {};
    \node (rn) at (6.1, -8) {};
    \path (ln) edge[double, ->]node[above,yshift=0.1em]{$\gamma_1$} (rn);
    \end{scope}
\end{tikzpicture}

%% file: sierpinsky-2-v2.tex
\input{sierpinsky-common-v2.tex}

\begin{tikzpicture}
    \begin{scope}[yscale=2/6*1.73]
        \begin{scope}[shift={(-1.3,0.7)}, scale=0.75]
            \ttri{0}{0}{0.75}{0.66}
            \draw [color=gray, rounded corners=5*0.66] (-0.5,-0.5) -- (4.5,-0.5) -- (2,7) --cycle;
        \end{scope}
        \begin{scope}[shift={(-0.7,-9)}, scale=0.4]
            \dttri{0}{0}{0.4}{1}
        \end{scope}
        \begin{scope}[shift={(8,0.7)}, scale=0.75]
            \tttri{0}{0}{0.75}{0.66}
            \draw [color=gray, rounded corners=5*0.66] (-0.5,-0.5) -- (4.5,-0.5) -- (2,7) --cycle;
        \end{scope}
        \begin{scope}[shift={(6,-9)}, scale=0.4]
            \dtttri{0}{0}{0.4}{1}
        \end{scope}
        
        \node (ut) at (-0.2, -.25*1.73) {};
        \node (dt) at (0.5, -1*1.73) {};
        \path (ut) edge[double, ->] node[below left,align=center]{(1) \it pattern\\\it matching}  (dt);
        
        \node (lt) at (3, 3) {};
        \node (rdt) at (6, 3) {};
        \path (lt) edge[thick, dashed, double, ->]node[above]{T} (rdt);
        
        \node (urdt) at (10, -0.25*1.73) {};
        \node (drdt) at (9.3, -1*1.73) {};
        \path (drdt) edge[double, ->]node[below right,align=left, xshift=0.1em]{(3) \it output\\ \it construction} (urdt);
        
        \node (rtt) at (4, -5) {};
        \node (ldtt) at (5, -5) {};
        \path (ldtt) edge[double, <-]node[above,yshift=0.2em,align=center]{(2) \it local\\\it application} (rtt);

    \end{scope}
\end{tikzpicture}